\documentclass[a4paper]{article}%
\usepackage{amsmath}
\usepackage{amsfonts}
\usepackage{amssymb}
\usepackage{graphicx}%
\providecommand{\U}[1]{\protect\rule{.1in}{.1in}}
\newtheorem{theorem}{Theorem}

\newtheorem{lemma}[theorem]{Lemma}

\newtheorem{remark}[theorem]{Remark}

\newenvironment{proof}[1][Proof]{\noindent\textbf{#1.} }{\ \rule{0.5em}{0.5em}}
\begin{document}

\title{\textbf{Empirical phi-divergence test-statistics for the equality of means of
two populations}}
\author{N. Balakrishnan$^{1}$, N. Mart\'{\i}n$^{2}$ and L. Pardo$^{3}$\\$^{1}${\small Department of Mathematics and Statistics, McMaster University,
Hamilton, Canada}\\$^{2}${\small Department of Statistics and O.R. II (Decision Methods),
Complutense University of Madrid, 28003 Madrid, Spain}\\$^{3}${\small Department of Statistics and O.R., Complutense University of
Madrid, 28040 Madrid, Spain}}
\maketitle

\begin{abstract}
Empirical phi-divergence test-statistics have demostrated to be a useful
technique for the simple null hypothesis to improve the finite sample
behaviour of the classical likelihood ratio test-statistic, as well as for
model misspecification problems, in both cases for the one population problem.
This paper introduces this methodology for two sample problems. A simulation
study illustrates situations in which the new test-statistics become a
competitive tool with respect to the classical z-test and the likelihood ratio test-statistic.

\end{abstract}

\bigskip\bigskip

\noindent\textbf{AMS 2001 Subject Classification: }62F03, 62F25.

\noindent\textbf{Keywords and phrases: }Empirical likelihood, Empirical
phi-divergence test statistics, Phi-divergence measures, Power function.

\section{Introduction\label{sec1}}

The method of likelihood introduced by Fisher is certainly one of the most
commonly used techniques for parametric models. The likelihood has been also
shown to be very useful in non-parametric context. More concretely Owen (1988,
1990, 1991) introduced the empirical likelihood ratio statistics for
non-parametric problems. Two sample problems are frequently encountered in
many areas of statistics, generally performed under the assumption of
normality. The most commonly used test in this connection is the two sample
$t$-test for the equality of means, performed under the assumption of equality
of variances. If the variances are unknown, we have the so-called
Behrens-Fisher problem. It is well-known that the two sample $t$-test has cone
major drawback; it is highly sensitive to deviations from the ideal
conditions, and may perform miserably under model misspecification and the
presence of outliers. Recently Basu et al. (2014) presented a new family of
test statistics to overcome the problem of non-robustness of the $t$-statistic.

Empirical likelihood methods for two-sample problems have been studied by
different researchers since Owen (1988) introduced the empirical likelihood as
a non-parametric likelihood-based alternative approach to inference on the
mean of a single population. The monograph of Owen (2001) is an excellent
overview of developments on empirical likelihood and considers a multi-sample
empirical likelihood theorem, which includes the two-sample problem as a
special case. Some important contributions for the two-sample problem are
given in Owen (1991), Adimiri (1995), Jin (1995), Qin (1994, 1998), Qin and
Zhao (2000), Zhang (2000), Liu et al. (2008), Baklizi and Kibria (2009), Wu
and Yan (2012) and references therein.

Consider two independent unidimensional random variables $X$ with unknown mean
$\mu_{1}$ and variance $\sigma_{1}^{2}$ and $Y$ with unknown mean $\mu_{2}$
and variance $\sigma_{2}^{2}$. Let $X_{1},...,X_{m}$ be a random sample of
size $m$ from the population denoted by $X$, with distribution function $F$,
and $Y_{1},...,Y_{n}$ be a random sample of size $n$ from the population
denoted by $Y$, with distribution function $G$. We shall assume that $F$ and
$G$ are unknown, therefore we are interested in a non-parametric approach,
more concretely we shall use empirical likelihood methods. If we denote
$\mu_{1}=\mu$ and $\mu_{2}=\mu+\delta$, our interest will be in testing
\begin{equation}
H_{0}\text{: }\delta=\delta_{0}\text{ vs. }H_{1}\text{: }\delta\neq\delta_{0},
\label{1}%
\end{equation}
being $\delta_{0}$ a known real number. Since $\delta=\mu_{2}-\mu_{1}$ becomes
the parameter of interest, apart from testing (\ref{1}), we might also be
interested in constructing the confidence interval for $\delta$.

In this paper we are going to introduce a new family of empirical test
statistics for the two-sample problem introduced in (\ref{1}): Empirical
phi-divergence test statistics. This family of test statistics is based on
phi-divergence measures and it contains the empirical log-likelihood ratio
test statistic as a particular case. In this sense, we can think that the
family of empirical phi-divergence test statistics presented and studied in
this paper is a generalization of the empirical log-likelihood ratio statistic.

Let $N=m+n$, assume that%
\begin{equation}
\frac{m}{N}\underset{m,n\rightarrow\infty}{\longrightarrow}\nu\in\left(
0,1\right)  , \label{ass}%
\end{equation}
and $x_{1},...,x_{m}$, $y_{1},...,y_{n}$ a realization of $X_{1},...,X_{m}$,
$Y_{1},...,Y_{n}$. We denote%
\[
\mathcal{L}\left(  \delta\right)  =%
{\textstyle\prod\limits_{i=1}^{m}}
{\textstyle\prod\limits_{j=1}^{n}}
p_{i}q_{j}\text{ s.t. }%
{\textstyle\sum\limits_{i=1}^{m}}
p_{i}(x_{i}-\mu)=0\text{, }%
{\textstyle\sum\limits_{i=1}^{m}}
p_{i}=1\text{, }%
{\textstyle\sum\limits_{j=1}^{n}}
q_{j}(y_{j}-\mu-\delta)=0\text{, }%
{\textstyle\sum\limits_{j=1}^{n}}
q_{j}=1\text{,}%
\]
and
\[
\mathcal{L}\left(  \boldsymbol{p},\boldsymbol{q}\right)  =%
{\textstyle\prod\limits_{i=1}^{m}}
{\textstyle\prod\limits_{j=1}^{n}}
p_{i}q_{j}\text{ s.t. }%
{\textstyle\sum\limits_{i=1}^{m}}
p_{i}=%
{\textstyle\sum\limits_{j=1}^{n}}
q_{j}=1\text{, }p_{i},q_{j}\geq0\text{,}%
\]
with $p_{i}=p_{i}(\mu)=F\left(  x_{i}\right)  -F\left(  x_{i}^{-}\right)  $,
$q_{j}=q_{j}(\mu,\delta)=G\left(  y_{j}\right)  -G(y_{j}^{-})$ and
$\boldsymbol{p}=\left(  p_{1},...,p_{m}\right)  ^{T}$, $\boldsymbol{q}=\left(
q_{1},...,q_{n}\right)  ^{T}$.

The empirical log-likelihood ratio statistic for testing (\ref{1}) is given
by
\begin{equation}
\mathcal{\ell}\left(  \delta_{0}\right)  =-2\log\frac{\sup_{\boldsymbol{p}%
,\boldsymbol{q}}\mathcal{L}\left(  \delta_{0}\right)  }{\sup_{\boldsymbol{p}%
,\boldsymbol{q}}\mathcal{L}\left(  \boldsymbol{p},\boldsymbol{q}\right)  }.
\label{C}%
\end{equation}
Using the standard Lagrange multiplier method we might obtain $\sup
_{\boldsymbol{p},\boldsymbol{q}}\mathcal{L}\left(  \delta_{0}\right)  $, as
well as $\sup_{\boldsymbol{p},\boldsymbol{q}}\mathcal{L}\left(  \boldsymbol{p}%
,\boldsymbol{q}\right)  $. For $\sup_{\boldsymbol{p},\boldsymbol{q}%
}\mathcal{L}\left(  \delta_{0}\right)  $, taking derivatives on
\[
\mathcal{L}_{1}=%
{\textstyle\sum\limits_{i=1}^{m}}
\log p_{i}+%
{\textstyle\sum\limits_{j=1}^{n}}
\log q_{j}+s_{1}\left(  1-%
{\textstyle\sum\limits_{i=1}^{m}}
p_{i}\right)  +s_{2}\left(  1-%
{\textstyle\sum\limits_{j=1}^{n}}
q_{j}\right)  -\lambda_{1}m%
{\textstyle\sum\limits_{i=1}^{m}}
p_{i}(x_{i}-\mu)-\lambda_{2}n%
{\textstyle\sum\limits_{j=1}^{n}}
q_{j}(y_{j}-\mu-\delta_{0}),
\]
we obtain
\begin{align}
\frac{\partial\mathcal{L}_{1}}{\partial p_{i}}  &  =0\Leftrightarrow
p_{i}=\frac{1}{m}\frac{1}{1+\lambda_{1}\left(  x_{i}-\mu\right)  },\text{
}i=1,...,m,\label{2}\\
\frac{\partial\mathcal{L}_{1}}{\partial q_{j}}  &  =0\Leftrightarrow
q_{j}=\frac{1}{n}\frac{1}{1+\lambda_{2}\left(  y_{j}-\mu-\delta_{0}\right)
},\text{ }j=1,...,n, \label{3}%
\end{align}
and
\[
\frac{\partial\mathcal{L}_{1}}{\partial\mu}=0\Leftrightarrow m\lambda
_{1}+n\lambda_{2}=0.
\]
Therefore, the empirical maximum likelihood estimates $\widetilde{\lambda}%
_{1}$, $\widetilde{\lambda}_{2}$ and $\widetilde{\mu}$ of $\lambda_{1}$,
$\lambda_{2}$ and $\mu$, under $H_{0}$, are obtained as the solution of the
equations%
\begin{equation}
\left\{
\begin{array}
[c]{l}%
\dfrac{1}{m}%
{\textstyle\sum\limits_{i=1}^{m}}
\frac{1}{1+\lambda_{1}\left(  x_{i}-\mu\right)  }=1\\
\dfrac{1}{n}%
{\textstyle\sum\limits_{j=1}^{n}}
\frac{1}{1+\lambda_{2}\left(  y_{j}-\mu-\delta_{0}\right)  }=1\\
m\lambda_{1}+n\lambda_{2}=0
\end{array}
\right.  , \label{3bis}%
\end{equation}
and%
\begin{equation}
\log\sup_{\boldsymbol{p},\boldsymbol{q}}\mathcal{L}\left(  \delta_{0}\right)
=-m\log m-%
{\textstyle\sum\limits_{i=1}^{m}}
\log\left(  1+\widetilde{\lambda}_{1}\left(  x_{i}-\widetilde{\mu}\right)
\right)  -n\log n-%
{\textstyle\sum\limits_{j=1}^{n}}
\log\left(  1+\widetilde{\lambda}_{2}\left(  y_{j}-\widetilde{\mu}-\delta
_{0}\right)  \right)  . \label{4}%
\end{equation}
In relation $\sup_{\boldsymbol{p},\boldsymbol{q}}\mathcal{L}\left(
\boldsymbol{p},\boldsymbol{q}\right)  $, taking derivatives on%
\begin{equation}
\mathcal{L}_{2}=%
{\textstyle\sum\limits_{i=1}^{m}}
\log p_{i}+%
{\textstyle\sum\limits_{j=1}^{n}}
\log q_{j}+s_{1}\left(  1-%
{\textstyle\sum\limits_{i=1}^{m}}
p_{i}\right)  +s_{2}\left(  1-%
{\textstyle\sum\limits_{j=1}^{n}}
q_{j}\right)  , \label{5}%
\end{equation}
we have
\begin{equation}
\frac{\partial\mathcal{L}_{2}}{\partial p_{i}}=0\Leftrightarrow p_{i}=\frac
{1}{m}\text{, }i=1,...,m\text{ and }\frac{\partial\mathcal{L}_{2}}{\partial
q_{j}}=0\Leftrightarrow q_{j}=\frac{1}{n}\text{, }j=1,...,n, \label{6}%
\end{equation}
and%
\begin{equation}
\log\sup_{\boldsymbol{p},\boldsymbol{q}}\mathcal{L}\left(  \boldsymbol{p}%
,\boldsymbol{q}\right)  =-m\log m-n\log n. \label{7}%
\end{equation}
Therefore, the empirical log-likelihood ratio statistic (\ref{C}), for testing
(\ref{1}), can be written as
\begin{align}
\mathcal{\ell}\left(  \delta_{0}\right)   &  =-2\left(  -m\log m-%
{\textstyle\sum\limits_{i=1}^{m}}
\log\left(  1+\widetilde{\lambda}_{1}\left(  X_{i}-\widetilde{\mu}\right)
\right)  -n\log n\right. \nonumber\\
&  \left.  -%
{\textstyle\sum\limits_{j=1}^{n}}
\log\left(  1+\widetilde{\lambda}_{2}\left(  Y_{j}-\widetilde{\mu}-\delta
_{0}\right)  \right)  +m\log m+n\log n\right) \nonumber\\
&  =2\left\{
{\textstyle\sum\limits_{i=1}^{m}}
\log\left(  1+\widetilde{\lambda}_{1}\left(  X_{i}-\widetilde{\mu}\right)
\right)  +%
{\textstyle\sum\limits_{j=1}^{n}}
\log\left(  1+\widetilde{\lambda}_{2}\left(  Y_{j}-\widetilde{\mu}-\delta
_{0}\right)  \right)  \right\}  . \label{8}%
\end{align}
Under some regularity conditions, Jing (1995) established that
\[
\Pr\left(  \mathcal{\ell}\left(  \delta_{0}\right)  >\chi_{1,\alpha}%
^{2}\right)  =\alpha+O(n^{-1}),
\]
where $\chi_{1,\alpha}^{2}$ is the $100(1-\alpha)$-th percentile of the
$\chi_{1}^{2}$ distribution.

Our interest in this paper is to study the problem of testing given in
(\ref{1}) and at the same time to construct confidence intervals for $\delta$
on the basis of the empirical phi-divergence test statistics. Empirical
phi-divergence test statistics in the context of the empirical likelihood have
studied by Baggerly (1998), Broniatowski and Keizou (2012), Balakhrishnan et
al. (2013), Felipe et al. (2015) and references therein. The family of
empirical phi-divergence test statistics, considered in this paper, contains
the classical empirical log-likelihood ratio statistic as a particular case.
In Section \ref{sec2}, the empirical phi-divergence test statistics are
introduced and the corresponding asymptotic distributions are obtained. A
simulation study is carried out in Section \ref{sec4}. Section \ref{sec3} is
devoted to develop a numerical example. In Section \ref{sec5} the previous
results, devoted to univariate populations, are extended to $k$-dimensional populations.

\section{Empirical phi-divergence test statistics\label{sec2}}

For the hypothesis testing considered in (\ref{1}), in this section the family
of empirical phi-divergence test statistics are introduced as a natural
extension of the empirical log-likelihood ratio statistic given in (\ref{C}).

We consider the $N$-dimensional probability vectors
\begin{equation}
\boldsymbol{U}=(\tfrac{1}{N},\overset{\underset{\smile}{N}}{...},\tfrac{1}%
{N})^{T} \label{9}%
\end{equation}
and
\begin{equation}
\boldsymbol{P=}\left(  p_{1}\nu,...,p_{m}\nu,q_{1}(1-\nu),...,q_{n}%
(1-\nu)\right)  ^{T} \label{10}%
\end{equation}
where $p_{i},$ $i=1,...,m$, $q_{j},$ $j=1,...,n$ were defined in (\ref{2}) and
(\ref{3}), respectively, and $\nu$\ in (\ref{ass}). Let
$\widetilde{\boldsymbol{P}}$ be the $N$-dimensional vector obtained from
$\boldsymbol{P}$ with $p_{i}$, $q_{j}$\ replaced by the corresponding
empirical maximum likelihood estimators $\widetilde{p}_{i}$, $\widetilde{q}%
_{j}$ and $\nu$ by $\frac{m}{N}$. The Kullback-Leibler divergence between the
probability vectors $\boldsymbol{U}$ and $\widetilde{\boldsymbol{P}}$ is given
by
\begin{align*}
D_{Kullback}(\boldsymbol{U},\widetilde{\boldsymbol{P}})  &  =%
{\textstyle\sum\limits_{i=1}^{m}}
\frac{1}{N}\log\frac{\frac{1}{N}}{\widetilde{p}_{i}\frac{m}{N}}+%
{\textstyle\sum\limits_{j=1}^{n}}
\frac{1}{N}\log\frac{\frac{1}{N}}{\widetilde{q}_{j}\left(  1-\frac{m}%
{N}\right)  }\\
&  =-\frac{1}{N}\left\{
{\textstyle\sum\limits_{i=1}^{m}}
\log m\widetilde{p}+%
{\textstyle\sum\limits_{j=1}^{n}}
\log n\widetilde{q}_{j}\right\} \\
&  =\frac{1}{N}\left\{
{\textstyle\sum\limits_{i=1}^{m}}
\log\left(  1+\widetilde{\lambda}_{1}\left(  x_{i}-\widetilde{\mu}\right)
\right)  +%
{\textstyle\sum\limits_{j=1}^{n}}
\log\left(  1+\widetilde{\lambda}_{2}\left(  y_{j}-\widetilde{\mu}-\delta
_{0}\right)  \right)  \right\}  ,
\end{align*}
where%
\begin{align}
\widetilde{p}_{i}  &  =\frac{1}{m}\frac{1}{1+\widetilde{\lambda}_{1}\left(
x_{i}-\widetilde{\mu}\right)  },\text{ }i=1,...,m,\label{p}\\
\widetilde{q}_{j}  &  =\frac{1}{n}\frac{1}{1+\widetilde{\lambda}_{2}\left(
y_{j}-\widetilde{\mu}-\delta_{0}\right)  },\text{ }j=1,...,n. \label{q}%
\end{align}
Therefore, the relationship between $\mathcal{\ell}\left(  \delta_{0}\right)
$ and $D_{Kullback}(\boldsymbol{U},\widetilde{\boldsymbol{P}})$ is%
\begin{equation}
\mathcal{\ell}\left(  \delta_{0}\right)  =2ND_{Kullback}(\boldsymbol{U}%
,\widetilde{\boldsymbol{P}}). \label{11}%
\end{equation}
Based on (\ref{11}), in this paper the empirical phi-divergence test
statistics for (\ref{1}) are introduced for the first time. This family of
empirical phi-divergence test statistics is obtained replacing the
Kullback-Leibler divergence by a phi-divergence measure in (\ref{11}), i.e.,
\begin{equation}
T_{\phi}\left(  \delta_{0}\right)  =\frac{2N}{\phi^{\prime\prime}(1)}D_{\phi
}(\boldsymbol{U},\widetilde{\boldsymbol{P}}), \label{11BIS}%
\end{equation}
where
\[
D_{\phi}(\boldsymbol{U},\widetilde{\boldsymbol{P}})=%
{\textstyle\sum\limits_{i=1}^{m}}
\widetilde{p}_{i}\frac{m}{N}\phi\left(  \frac{\frac{1}{N}}{\widetilde{p}%
_{i}\frac{m}{N}}\right)  +%
{\textstyle\sum\limits_{j=1}^{n}}
\left(  1-\frac{m}{N}\right)  \widetilde{q}_{j}\phi\left(  \frac{\frac{1}{N}%
}{\left(  1-\frac{m}{N}\right)  \widetilde{q}_{j}}\right)  ,
\]
with $\phi:\mathbb{R}^{+}\longrightarrow\mathbb{R}$ being any convex function
such that at $x=1$, $\phi\left(  1\right)  =0$, $\phi^{\prime\prime}\left(
1\right)  >0$ and at $x=0$, $0\phi\left(  0/0\right)  =0$ and $0\phi\left(
p/0\right)  =p\lim_{u\rightarrow\infty}\frac{\phi\left(  u\right)  }{u}$. For
more details see Cressie and Pardo (2002) and Pardo (2006). Therefore,
(\ref{11BIS}) can be rewritten as%
\begin{align}
T_{\phi}\left(  \delta_{0}\right)   &  =\frac{2}{\phi^{\prime\prime}%
(1)}\left\{
{\textstyle\sum\limits_{i=1}^{m}}
m\widetilde{p}_{i}\phi\left(  \frac{1}{m\widetilde{p}_{i}}\right)  +%
{\textstyle\sum\limits_{j=1}^{n}}
n\widetilde{q}_{j}\phi\left(  \frac{1}{n\widetilde{q}_{j}}\right)  \right\}
\label{12}\\
&  =\frac{2}{\phi^{\prime\prime}(1)}\left\{
{\textstyle\sum\limits_{i=1}^{m}}
\frac{1}{1+\widetilde{\lambda}_{1}\left(  x_{i}-\widetilde{\mu}\right)  }%
\phi\left(  1+\widetilde{\lambda}_{1}\left(  x_{i}-\widetilde{\mu}\right)
\right)  +%
{\textstyle\sum\limits_{j=1}^{n}}
\frac{1}{1+\widetilde{\lambda}_{2}\left(  y_{j}-\widetilde{\mu}-\delta
_{0}\right)  }\phi\left(  1+\widetilde{\lambda}_{2}\left(  y_{j}-\left(
\widetilde{\mu}+\delta_{0}\right)  \right)  \right)  \right\}  .\nonumber
\end{align}
If $\phi(x)=x\log x-x+1$ is chosen in $D_{\phi}(\boldsymbol{U}%
,\widetilde{\boldsymbol{P}})$, we get the Kullback-Leibler divergence and
$T_{\phi}\left(  \delta_{0}\right)  $ coincides with the empirical
log-likelihood ratio statistic $\mathcal{\ell}\left(  \delta_{0}\right)  $
given in (\ref{11}).

Let $\widehat{\mu}(\sigma_{1}^{2},\sigma_{2}^{2})$ be the optimal estimator of
$\mu$ under the assumption of having the known values of $\sigma_{1}^{2}$,
$\sigma_{2}^{2}$, i.e. it is given by the shape $\pi\overline{X}%
+(1-\pi)\overline{Y}$ and has minimum variance. It is well-known that%
\begin{equation}
\widehat{\mu}(\sigma_{1}^{2},\sigma_{2}^{2})=\dfrac{\dfrac{m\overline{X}%
}{\sigma_{1}^{2}}+n\dfrac{\left(  \overline{Y}-\delta_{0}\right)  }{\sigma
_{2}^{2}}}{\dfrac{m}{\sigma_{1}^{2}}+\dfrac{n}{\sigma_{2}^{2}}}. \label{muHat}%
\end{equation}
Similarly, an asymptotically optimal estimator of $\mu$ having unknown values
of $\sigma_{1}^{2}$, $\sigma_{2}^{2}$, is given by%
\[
\widehat{\mu}(S_{1}^{2},S_{2}^{2})=\dfrac{\dfrac{m\overline{X}}{S_{1}^{2}%
}+n\dfrac{\left(  \overline{Y}-\delta_{0}\right)  }{S_{1}^{2}}}{\dfrac
{m}{S_{1}^{2}}+\dfrac{n}{S_{1}^{2}}},
\]
where $S_{1}^{2}=\frac{1}{m-1}\sum_{i=1}^{m}(X_{i}-\overline{X})^{2}$,
$S_{2}^{2}=\frac{1}{n-1}\sum_{j=1}^{n}(Y_{j}-\overline{Y})^{2}$ are consistent
estimators of $\sigma_{1}^{2}$, $\sigma_{2}^{2}$ respectively. In the
following lemma an important relationship is established, useful to get the
asymptotic distribution of $T_{\phi}\left(  \delta_{0}\right)  $.

\begin{lemma}
\label{Lem}Let $\widetilde{\mu}$ the empirical likelihood estimator of $\mu$.
Then, we have
\[
\widetilde{\mu}=\widehat{\mu}(\sigma_{1}^{2},\sigma_{2}^{2})+O_{p}%
(1)=\widehat{\mu}(S_{1}^{2},S_{2}^{2})+O_{p}(1).
\]

\end{lemma}

\begin{proof}
See Appendix \ref{Ap1}.
\end{proof}

\begin{theorem}
\label{Th1}Suppose that $\sigma_{1}^{2}<\infty$, $\sigma_{2}^{2}<\infty$ and
(\ref{ass}). Then,%
\[
T_{\phi}\left(  \delta_{0}\right)  \underset{n,m\rightarrow\infty
}{\overset{\mathcal{L}}{\rightarrow}}\chi_{1}^{2}.
\]

\end{theorem}

\begin{proof}
See Appendix \ref{Ap2}.
\end{proof}

\begin{remark}
\label{R3}A $(1-\alpha)$-level confidence interval on $\delta$ can be
constructed as
\[
CI_{1-\alpha}(\delta)=\left\{  \delta:T_{\phi}\left(  \delta\right)  \leq
\chi_{1,\alpha}^{2}\right\}  .
\]
The lower and upper bounds of the interval $CI_{1-\alpha}(\delta)$ require a
bisection search algorithm. This is a computationally challenging task,
because for every selected grid point on $\delta$, one needs to maximize the
empirical phi-divergence $T_{\phi}\left(  \delta\right)  $ over the nuisance
parameter, $\mu$, and there is no closed-form solution to the maximum point
$\widehat{\mu}$ for any given $\delta$. The computational difficulties under
the standard two-sample empirical likelihood formulation are due to the fact
that the involved Lagrange multipliers, which are determined through the set
of equations (\ref{3bis}), have to be computed based on two separate samples
with an added nuisance parameter $\mu$. Such difficulties can be avoided
through an alternative formulation of the empirical likelihood function, for
which computation procedures are virtually identical to those for one-sample
of size $N=m+n$\ empirical likelihood problems. Through the transformations%
\begin{align*}
\boldsymbol{v}_{i}  &  =\left(  1-\omega_{1},\frac{x_{i}}{\omega_{1}}%
-\delta\right)  ^{T},\text{ }i=1,...,m,\\
\boldsymbol{w}_{j}  &  =\left(  -\omega_{1},\frac{y_{j}}{\omega_{2}}%
-\delta\right)  ^{T},\text{ }j=1,...,n,\\
\omega_{1}  &  =\omega_{1}=\frac{1}{2},
\end{align*}
(\ref{p}) and (\ref{q}) can be alternatively obtained as%
\begin{align}
\widetilde{p}_{i}  &  =\frac{1}{m}\frac{1}{1+\widetilde{\boldsymbol{\lambda}%
}_{\ast}^{T}\boldsymbol{v}_{i}},\text{ }i=1,...,m,\label{p2}\\
\widetilde{q}_{j}  &  =\frac{1}{n}\frac{1}{1+\widetilde{\boldsymbol{\lambda}%
}_{\ast}^{T}\boldsymbol{w}_{j}},\text{ }j=1,...,n, \label{q2}%
\end{align}
where the estimates of the Lagrange multipliers
$\widetilde{\boldsymbol{\lambda}}_{\ast}\boldsymbol{=}(\widetilde{\lambda
}_{1,\ast},\widetilde{\lambda}_{2,\ast})^{T}$ are the solution in
$\boldsymbol{\lambda}_{\ast}$\ of
\[%
{\textstyle\sum\limits_{i=1}^{m}}
\frac{\boldsymbol{v}_{i}}{1+\boldsymbol{\lambda}_{\ast}^{T}\boldsymbol{v}_{i}%
}+%
{\textstyle\sum\limits_{j=1}^{n}}
\frac{\boldsymbol{w}_{j}}{1+\boldsymbol{\lambda}_{\ast}^{T}\boldsymbol{w}_{j}%
}=\boldsymbol{0}_{2}.
\]

\end{remark}

\begin{remark}
\label{R4}In the particular case that $m=n$, the two samples might be
understood as a random sample of size $n$ from a unique bidimensional
population. In this setting the two sample problem can be considered to be a
particular case of Balakrishnan et al. (2015).
\end{remark}

\begin{remark}
\label{R5}Fu et al. (2009), Yan (2010) and Wu and Yan (2012) pointed out that
empirical log-likelihood ratio statistic, $\mathcal{\ell}\left(  \delta
_{0}\right)  $, given in (\ref{8}) for testing (\ref{1}), does not perform
well when the distribution associated to the samples are quite skewed or
samples sizes are not large or sample sizes from each population are quite
different. To overcome this problem Fu et al. (2009) considered the weighted
empirical log-likelihood function defined by%
\begin{equation}
\mathcal{\ell}_{w}\left(  \boldsymbol{p},\boldsymbol{q}\right)  =\frac
{\omega_{1}}{m}%
{\textstyle\sum\limits_{i=1}^{m}}
\log p_{i}+\frac{\omega_{2}}{n}%
{\textstyle\sum\limits_{j=1}^{n}}
\log q_{j}, \label{r1}%
\end{equation}
with $\omega_{1}=\omega_{2}=\frac{1}{2}$, and obtained the weighted empirical
likelihood (WEL) estimator as well as the weighted empirical log-likelihood
ratio statistic. In order to get the WEL estimator, it is necessary to
maximize (\ref{r1}) subject to%
\begin{align}
&
{\textstyle\sum\limits_{i=1}^{m}}
p_{i}=%
{\textstyle\sum\limits_{j=1}^{n}}
q_{j}=1,\label{r2}\\
&
{\textstyle\sum\limits_{i=1}^{m}}
p_{i}x_{i}-%
{\textstyle\sum\limits_{j=1}^{n}}
y_{j}q_{j}=\delta_{0}. \label{r3}%
\end{align}
They obtained that the WEL estimates of $p_{i}$ and $q_{j}$ are given by%
\begin{align*}
^{w}\widetilde{p}_{i}  &  =\frac{1}{m}\frac{1}{1+\;^{w}%
\!\widetilde{\boldsymbol{\lambda}}^{T}\boldsymbol{v}_{i}},\text{ }i=1,...,m,\\
^{w}\widetilde{q}_{j}  &  =\frac{1}{n}\frac{1}{1+\;^{w}%
\!\widetilde{\boldsymbol{\lambda}}^{T}\boldsymbol{w}_{j}},\text{ }j=1,...,n,
\end{align*}
where $\boldsymbol{v}_{i}$ and $\boldsymbol{w}_{j}$ are the same
transformations given in Remark \ref{R3} with $\delta=\delta_{0}$\ and the
estimates of the Lagrange multipliers $\;^{w}\!\widetilde{\boldsymbol{\lambda
}}_{\ast}\boldsymbol{=}(\;^{w}\!\widetilde{\lambda}_{1,\ast},\;^{w}%
\!\widetilde{\lambda}_{2,\ast})^{T}$ are the solution in $\;^{w}%
\!\boldsymbol{\lambda}_{\ast}$\ of
\[
\frac{\omega_{1}}{m}%
{\textstyle\sum\limits_{i=1}^{m}}
\frac{\boldsymbol{v}_{i}}{1+\;^{w}\!\boldsymbol{\lambda}_{\ast}^{T}%
\boldsymbol{v}_{i}}+\frac{\omega_{2}}{n}%
{\textstyle\sum\limits_{j=1}^{n}}
\frac{\boldsymbol{w}_{j}}{1+\;^{w}\!\boldsymbol{\lambda}_{\ast}^{T}%
\boldsymbol{w}_{j}}=\boldsymbol{0}_{2}.
\]
Now, if we define the probability vectors
\begin{align*}
^{w}\boldsymbol{U}  &  \boldsymbol{=}\left(  \omega_{1}(\tfrac{1}%
{m},\overset{\underset{\smile}{m}}{...},\tfrac{1}{m}),\omega_{2}(\tfrac{1}%
{n},,\overset{\underset{\smile}{n}}{...},,\tfrac{1}{n})\right)  ^{T},\\
^{w}\boldsymbol{P}  &  =\left(  \omega_{1}\boldsymbol{p}^{T},\omega
_{2}\boldsymbol{q}^{T}\right)  ^{T}=\left(  \omega_{1}\left(  p_{1}%
,...,p_{m}\right)  ,\omega_{2}(q_{1},...,q_{n})\right)  ^{T},
\end{align*}
the weighted empirical log-likelihood ratio test $\mathcal{\ell}_{w}\left(
\delta_{0}\right)  $ presented in Wu and Yan (2012) can be written as
\begin{equation}
-2\mathcal{\ell}_{w}\left(  \delta_{0}\right)  =D_{Kullback}(^{w}%
\boldsymbol{U},^{w}\boldsymbol{P}). \label{r4}%
\end{equation}
The weighted empirical log-likelihood ratio test can be extended by defining
the family of weighted empirical phi-divergence test statistics as
\[
S_{\phi}\left(  \delta_{0}\right)  =\frac{2D_{\phi}(^{w}\boldsymbol{U}%
,^{w}\boldsymbol{P})}{\phi^{\prime\prime}(1)},
\]
where $D_{\phi}(^{w}\boldsymbol{U},^{w}\boldsymbol{P})$ is the phi-divergence
measure between the probability vectors $^{w}\boldsymbol{U}$ and
$^{w}\boldsymbol{P}$, i.e.,%
\begin{align*}
D_{\phi}(^{w}\boldsymbol{U},^{w}\boldsymbol{P})  &  =\omega_{1}%
{\textstyle\sum\limits_{i=1}^{m}}
\widetilde{p}_{i}\phi\left(  \frac{1}{m\widetilde{p}_{i}}\right)  +\omega_{2}%
{\textstyle\sum\limits_{j=1}^{n}}
\widetilde{q}_{j}\phi\left(  \frac{1}{n\widetilde{q}_{j}}\right) \\
&  =\frac{\omega_{1}}{m}%
{\textstyle\sum\limits_{i=1}^{m}}
\frac{\phi\left(  1+\;^{w}\!\widetilde{\boldsymbol{\lambda}}_{\ast}%
^{T}\boldsymbol{v}_{i}\right)  }{1+\;^{w}\!\widetilde{\boldsymbol{\lambda}%
}_{\ast}^{T}\boldsymbol{v}_{i}}+\frac{\omega_{2}}{n}%
{\textstyle\sum\limits_{j=1}^{n}}
\frac{\phi\left(  1+\;^{w}\!\widetilde{\boldsymbol{\lambda}}_{\ast}%
^{T}\boldsymbol{w}_{j}\right)  }{1+\;^{w}\!\widetilde{\boldsymbol{\lambda}%
}_{\ast}^{T}\boldsymbol{w}_{j}}.
\end{align*}
Taking into account
\[
D_{\phi}(^{w}\boldsymbol{U},^{w}\boldsymbol{P})=\;^{w}%
\!\widetilde{\boldsymbol{\lambda}}_{\ast}^{T}\boldsymbol{D}\;^{w}%
\!\widetilde{\boldsymbol{\lambda}}_{\ast}\boldsymbol{\boldsymbol{+}}o_{p}(N),
\]
where%
\[
\boldsymbol{D}=\frac{\omega_{1}}{m}%
{\textstyle\sum\limits_{i=1}^{m}}
\boldsymbol{v}_{i}\boldsymbol{v}_{i}^{T}\boldsymbol{+}\frac{\omega_{2}}{n}%
{\textstyle\sum\limits_{j=1}^{n}}
\boldsymbol{w}_{j}\boldsymbol{w}_{j}^{T},
\]
and based on Theorem 2.2. in Wu and Yan (2012), we have that%
\[
\frac{S_{\phi}\left(  \delta_{0}\right)  }{c}\overset{\mathcal{L}%
}{\underset{n,m\rightarrow\infty}{\longrightarrow}}\chi_{1}^{2},
\]
where $c$ is the second diagonal element of the matrix $\boldsymbol{D}^{-1}$.
\end{remark}

\section{Simulation Study\label{sec4}}

The square of the classical $z$-test statistic for two sample problems,%
\[
t\left(  \delta_{0}\right)  =\frac{\left(  \overline{X}-\overline{Y}%
+\delta_{0}\right)  ^{2}}{\frac{1}{m}S_{1}^{2}+\frac{1}{n}S_{2}^{2}},
\]
has asymptotically $\chi_{1}^{2}$ distribution, the same as the empirical
phi-divergence test statistics, according to Theorem \ref{Th1}. In order to
compare the finite sample performance of the confidence interval (CI) of
$\delta$\ based on $T_{\phi}\left(  \delta\right)  $ with respect to the ones
based on $t\left(  \delta\right)  $ as well as the empirical log-likelihood
ratio test-statistic $\mathcal{\ell}\left(  \delta\right)  $ given in
(\ref{C}),\ we count on a subfamily of phi-divergence measures, the so-called
power divergence measures $\phi_{\gamma}(x)=\frac{x^{1+\gamma}-x-\gamma
(x-1)}{\gamma(1+\gamma)}$, dependent of tuning parameter $\gamma\in%
\mathbb{R}
$, i.e.%
\[
T_{\gamma}\left(  \delta\right)  =\left\{
\begin{array}
[c]{ll}%
\dfrac{2}{\gamma(\gamma+1)}\left(
{\displaystyle\sum\limits_{i=1}^{m}}
(m\widetilde{p}_{i})^{-\gamma}+%
{\displaystyle\sum\limits_{j=1}^{n}}
(n\widetilde{q}_{j})^{-\gamma}-N\right)  , & \gamma\in%
\mathbb{R}
-\{0,-1\},\\
-2\left(  m\log m+n\log n+m%
{\displaystyle\sum\limits_{i=1}^{m}}
\log\widetilde{p}_{i}+n%
{\displaystyle\sum\limits_{j=1}^{n}}
\log\widetilde{q}_{j}\right)  , & \gamma=0,\\
2\left(  m\log m+n\log n+m%
{\displaystyle\sum\limits_{i=1}^{m}}
\widetilde{p}_{i}\log\widetilde{p}_{i}+n%
{\displaystyle\sum\limits_{j=1}^{n}}
\widetilde{q}_{j}\log\widetilde{q}_{j}\right)  , & \gamma=-1,
\end{array}
\right.
\]
where $\widetilde{p}_{i}$ and $\widetilde{q}_{j}$\ can be obtained from
(\ref{p2})-(\ref{q2}). We analyzed five new test-statistics, the empirical
power-divergence test statistics taking $\gamma\in\{-1,-0.5,\frac{2}{3}%
,1,2\}$. The case of $\gamma=0$ is not new, since the empirical log-likelihood
ratio test-statistic $\mathcal{\ell}\left(  \delta\right)  $ is a member of
the empirical power-divergence test statistics, i.e. $\mathcal{\ell}\left(
\delta\right)  =T_{\gamma=0}\left(  \delta\right)  $. The CI of $\delta
$\ based on $t\left(  \delta_{0}\right)  $\ with $100(1-\alpha)\%$ confidence
level is essentially the CI of $z$-test statistic, $(\widetilde{\delta}%
_{L},\widetilde{\delta}_{U})=(\overline{x}-\overline{y}-z_{\frac{\alpha}{2}%
}\sqrt{\frac{1}{m}s_{1}^{2}+\frac{1}{n}s_{2}^{2}},\overline{x}-\overline
{y}+z_{\frac{\alpha}{2}}\sqrt{\frac{1}{m}s_{1}^{2}+\frac{1}{n}s_{2}^{2}})$.
For $T\in\{T_{\gamma}\left(  \delta\right)  \}_{\gamma\in\Lambda}$,
$\Lambda=\{-1,-0.5,0,\frac{2}{3},1,2\}$, as mentioned in Remark \ref{R3},
since there is no explicit expression for $(\widetilde{\delta}_{L}%
,\widetilde{\delta}_{U})$ the bisection method should be followed. The
simulated coverage probabilities of the CI of $\delta$\ based on
$T\in\{t\left(  \delta\right)  \}\cup\{T_{\gamma}\left(  \delta\right)
\}_{\gamma\in\Lambda}$ were obtained with $R=15,000$ replications by%
\[
100\times\frac{1}{R}%
{\displaystyle\sum\limits_{r=1}^{R}}
\mathrm{I}(T^{(r)}\leq\chi_{1,\alpha}^{2}),
\]
with $\mathrm{I}(\cdot)$ being the indicator function. The simulated expected
width of the CI of $\delta$\ based on $T\in\{t\left(  \delta\right)
\}\cup\{T_{\gamma}\left(  \delta\right)  \}_{\gamma\in\Lambda}$ were obtained
with $R=3,000$ replications by%
\[
100\times\frac{1}{R}%
{\displaystyle\sum\limits_{r=1}^{R}}
(\widetilde{\delta}_{U}^{(r)}-\widetilde{\delta}_{L}^{(r)}).
\]
The reason why two different values of $R$ were followed is twofold. On one
hand calculating $\widetilde{\delta}_{U}^{(r)}-\widetilde{\delta}_{L}^{(r)}$
is much more time consuming than $\mathrm{I}(T^{(r)}\leq\chi_{1,\alpha}^{2})$
and on the other hand for the designed simulation experiment the replications
needed to obtain a good precision is less for the expected width than for the
coverage probability.

The simulation experiment is designed in a similar manner as in Wu and Yan
(2012). The true distributions, unknown in practice, are generated from:

\begin{description}
\item[i)] $X\sim\mathcal{N}(\mu,\sigma_{1}^{2})$, $Y\sim\mathcal{N}(\mu
+\delta_{0},\sigma_{2}^{2})$, with $\mu=1$, $\sigma_{1}^{2}=\sigma_{2}%
^{2}=1.5$, $\delta_{0}=0$;

\item[ii)] $X\sim\mathrm{lognormal}(\vartheta_{1},\theta_{1})$, $Y\sim
\mathrm{lognormal}(\vartheta_{2},\theta_{2})$, with $\vartheta_{1}=1.1$,
$\theta_{1}=0.4$, $\vartheta_{2}=1.2$, $\theta_{2}=0.2$.
\end{description}

Notice that in case ii) $\delta_{0}=0$ since $E[X]=E[Y]$. Depending on the
sample sizes, six scenarios were considered, $(m,n)\in
\{(15,30),(30,15),(30,30),(30,60),(60,30),(60,60)\}$. Table \ref{table1}
summarizes the results of the described simulation experiment with
$\alpha=0.05$. In all the cases and scenarios the narrower width is obtained
with $T_{\gamma=-1}\left(  \delta\right)  $, but the coverage probabilities
closest to $95\%$ depends on the case or scenario. For the case of the
lognormal distribution the CI based on $t\left(  \delta\right)  $
test-statistic has the closest coverage probability to $95\%$, but for the
case of the normal distribution $T_{\gamma=2/3}\left(  \delta\right)  $ and
$T_{\gamma=1}\left(  \delta\right)  $ power divergence based tend to have the
closest coverage probability to $95\%$.

In order to complement this study, the power functions have been drawn through
$R=15,000$ replications and taking $\delta$ as abscissa. For case i) the power
functions exhibit a symmetric shape with respect to the center and also a
parallel shape, in such a way that the test statistics with better
approximation of the size have worse power. For case ii), fixing the values of
the two parameters of $X$ and changing the two parameter of $Y$ as%
\[
\vartheta_{1}^{\prime}=k\vartheta_{1}^{\prime},\quad\theta_{1}^{\prime
}=k\theta_{1},\quad k=\frac{\log(\delta+\exp\{\vartheta_{1}+\frac{1}{2}%
\theta_{1}\})}{\vartheta_{2}+\frac{1}{2}\theta_{2}},
\]
$\delta$ is displaced from $\delta_{0}=0$ to the right when $k>1$ and from
$\delta_{0}=0$ to the left when $0<k<1$ ($\delta>-\exp\{\vartheta_{1}+\frac
{1}{2}\theta_{1}\}$). Unlike case i), the power function of case ii) exhibits
a different shape on both sides from the center of abscissa, and the most
prominent differences are on the left hand size. Clearly in case ii), even
though the approximated size for $t\left(  \delta_{0}\right)  $ is the best
one, it has the worst approximated power function, in particular there is an
area of the approximated power function on the left hand side of $\delta
_{0}=0$\ with smaller value than the approximated size. Hence, in case ii) the
power functions of $T\in\{T_{\gamma}\left(  \delta\right)  \}_{\gamma
\in\Lambda}$ are more acceptable than the power function of $t\left(
\delta_{0}\right)  $. Taking into account the strong and weak point of
$t\left(  \delta_{0}\right)  $ in case ii), $T_{\gamma=2/3}\left(  \delta
_{0}\right)  $ could be a good choice for moderate sample sizes and
$\mathcal{\ell}\left(  \delta_{0}\right)  =T_{\gamma=0}\left(  \delta
_{0}\right)  $ for small sample sizes.%

\begin{table}[htbp]  \renewcommand{\arraystretch}{0.84}
\tabcolsep2.8pt  \centering
$%
\begin{tabular}
[c]{ccccc}\hline
\multicolumn{5}{c}{case i): normal populations}\\\hline
$m$ & $n$ & $CI$ & coverage & width\\\hline
$15$ & $30$ & \multicolumn{1}{l}{$T_{\gamma=-1}\left(  \delta\right)  $} &
$92.1$ & \fbox{$1.56$}\\
$15$ & $30$ & \multicolumn{1}{l}{$T_{\gamma=-0.5}\left(  \delta\right)  $} &
$92.6$ & $1.60$\\
$15$ & $30$ & \multicolumn{1}{l}{$T_{\gamma=0}\left(  \delta\right)  $} &
$93.0$ & $1.64$\\
$15$ & $30$ & \multicolumn{1}{l}{$T_{\gamma=2/3}\left(  \delta\right)  $} &
$93.2$ & $1.65$\\
$15$ & $30$ & \multicolumn{1}{l}{$T_{\gamma=1}\left(  \delta\right)  $} &
$93.2$ & $1.65$\\
$15$ & $30$ & \multicolumn{1}{l}{$T_{\gamma=2}\left(  \delta\right)  $} &
$92.6$ & $1.63$\\
$15$ & $30$ & \multicolumn{1}{l}{$t\left(  \delta\right)  $} & \fbox{$93.5$} &
$1.66$\\\hline
$30$ & $15$ & \multicolumn{1}{l}{$T_{\gamma=-1}\left(  \delta\right)  $} &
$92.8$ & \fbox{$1.40$}\\
$30$ & $15$ & \multicolumn{1}{l}{$T_{\gamma=-0.5}\left(  \delta\right)  $} &
$93.3$ & $1.43$\\
$30$ & $15$ & \multicolumn{1}{l}{$T_{\gamma=0}\left(  \delta\right)  $} &
$93.6$ & $1.45$\\
$30$ & $15$ & \multicolumn{1}{l}{$T_{\gamma=2/3}\left(  \delta\right)  $} &
\fbox{$93.9$} & $1.46$\\
$30$ & $15$ & \multicolumn{1}{l}{$T_{\gamma=1}\left(  \delta\right)  $} &
\fbox{$93.9$} & $1.47$\\
$30$ & $15$ & \multicolumn{1}{l}{$T_{\gamma=2}\left(  \delta\right)  $} &
$93.6$ & $1.46$\\
$30$ & $15$ & \multicolumn{1}{l}{$t\left(  \delta\right)  $} & \fbox{$93.9$} &
$1.46$\\\hline
$30$ & $30$ & \multicolumn{1}{l}{$T_{\gamma=-1}\left(  \delta\right)  $} &
$93.8$ & \fbox{$1.24$}\\
$30$ & $30$ & \multicolumn{1}{l}{$T_{\gamma=-0.5}\left(  \delta\right)  $} &
$94.3$ & $1.29$\\
$30$ & $30$ & \multicolumn{1}{l}{$T_{\gamma=0}\left(  \delta\right)  $} &
$94.5$ & $1.28$\\
$30$ & $30$ & \multicolumn{1}{l}{$T_{\gamma=2/3}\left(  \delta\right)  $} &
$94.8$ & $1.30$\\
$30$ & $30$ & \multicolumn{1}{l}{$T_{\gamma=1}\left(  \delta\right)  $} &
\fbox{$94.9$} & $1.29$\\
$30$ & $30$ & \multicolumn{1}{l}{$T_{\gamma=2}\left(  \delta\right)  $} &
$94.7$ & $1.29$\\
$30$ & $30$ & \multicolumn{1}{l}{$t\left(  \delta\right)  $} & $94.7$ &
$1.28$\\\hline
$30$ & $60$ & \multicolumn{1}{l}{$T_{\gamma=-1}\left(  \delta\right)  $} &
$93.4$ & \fbox{$1.14$}\\
$30$ & $60$ & \multicolumn{1}{l}{$T_{\gamma=-0.5}\left(  \delta\right)  $} &
$93.8$ & $1.16$\\
$30$ & $60$ & \multicolumn{1}{l}{$T_{\gamma=0}\left(  \delta\right)  $} &
$94.1$ & $1.18$\\
$30$ & $60$ & \multicolumn{1}{l}{$T_{\gamma=2/3}\left(  \delta\right)  $} &
\fbox{$94.3$} & $1.20$\\
$30$ & $60$ & \multicolumn{1}{l}{$T_{\gamma=1}\left(  \delta\right)  $} &
\fbox{$94.3$} & $1.20$\\
$30$ & $60$ & \multicolumn{1}{l}{$T_{\gamma=2}\left(  \delta\right)  $} &
$94.1$ & $1.19$\\
$30$ & $60$ & \multicolumn{1}{l}{$t\left(  \delta\right)  $} & $94.2$ &
$1.18$\\\hline
$60$ & $30$ & \multicolumn{1}{l}{$T_{\gamma=-1}\left(  \delta\right)  $} &
$94.4$ & \fbox{$1.01$}\\
$60$ & $30$ & \multicolumn{1}{l}{$T_{\gamma=-0.5}\left(  \delta\right)  $} &
$94.6$ & $1.03$\\
$60$ & $30$ & \multicolumn{1}{l}{$T_{\gamma=0}\left(  \delta\right)  $} &
$94.8$ & $1.04$\\
$60$ & $30$ & \multicolumn{1}{l}{$T_{\gamma=2/3}\left(  \delta\right)  $} &
\fbox{$95.0$} & $1.04$\\
$60$ & $30$ & \multicolumn{1}{l}{$T_{\gamma=1}\left(  \delta\right)  $} &
\fbox{$95.0$} & $1.05$\\
$60$ & $30$ & \multicolumn{1}{l}{$T_{\gamma=2}\left(  \delta\right)  $} &
$94.9$ & $1.05$\\
$60$ & $30$ & \multicolumn{1}{l}{$t\left(  \delta\right)  $} & $94.8$ &
$1.04$\\\hline
$60$ & $60$ & \multicolumn{1}{l}{$T_{\gamma=-1}\left(  \delta\right)  $} &
$94.3$ & \fbox{$0.89$}\\
$60$ & $60$ & \multicolumn{1}{l}{$T_{\gamma=-0.5}\left(  \delta\right)  $} &
$94.5$ & $0.90$\\
$60$ & $60$ & \multicolumn{1}{l}{$T_{\gamma=0}\left(  \delta\right)  $} &
$94.7$ & $0.91$\\
$60$ & $60$ & \multicolumn{1}{l}{$T_{\gamma=2/3}\left(  \delta\right)  $} &
\fbox{$94.9$} & $0.92$\\
$60$ & $60$ & \multicolumn{1}{l}{$T_{\gamma=1}\left(  \delta\right)  $} &
\fbox{$94.9$} & $0.92$\\
$60$ & $60$ & \multicolumn{1}{l}{$T_{\gamma=2}\left(  \delta\right)  $} &
\fbox{$94.9$} & $0.92$\\
$60$ & $60$ & \multicolumn{1}{l}{$t\left(  \delta\right)  $} & $94.8$ &
$0.91$\\\hline
\end{tabular}
\ \ \ \ \qquad%
\begin{tabular}
[c]{ccccc}\hline
\multicolumn{5}{c}{case ii): lognormal populations}\\\hline
$m$ & $n$ & $CI$ & coverage & width\\\hline
$15$ & $30$ & \multicolumn{1}{l}{$T_{\gamma=-1}\left(  \delta\right)  $} &
$90.0$ & \fbox{$2.60$}\\
$15$ & $30$ & \multicolumn{1}{l}{$T_{\gamma=-0.5}\left(  \delta\right)  $} &
$90.6$ & $2.66$\\
$15$ & $30$ & \multicolumn{1}{l}{$T_{\gamma=0}\left(  \delta\right)  $} &
$91.0$ & $2.77$\\
$15$ & $30$ & \multicolumn{1}{l}{$T_{\gamma=2/3}\left(  \delta\right)  $} &
$90.8$ & $2.82$\\
$15$ & $30$ & \multicolumn{1}{l}{$T_{\gamma=1}\left(  \delta\right)  $} &
$90.6$ & $2.87$\\
$15$ & $30$ & \multicolumn{1}{l}{$T_{\gamma=2}\left(  \delta\right)  $} &
$89.2$ & $2.86$\\
$15$ & $30$ & \multicolumn{1}{l}{$t\left(  \delta\right)  $} & \fbox{$92.3$} &
$2.76$\\\hline
$30$ & $15$ & \multicolumn{1}{l}{$T_{\gamma=-1}\left(  \delta\right)  $} &
$92.2$ & \fbox{$2.35$}\\
$30$ & $15$ & \multicolumn{1}{l}{$T_{\gamma=-0.5}\left(  \delta\right)  $} &
$92.5$ & $2.46$\\
$30$ & $15$ & \multicolumn{1}{l}{$T_{\gamma=0}\left(  \delta\right)  $} &
$92.7$ & $2.52$\\
$30$ & $15$ & \multicolumn{1}{l}{$T_{\gamma=2/3}\left(  \delta\right)  $} &
$92.5$ & $2.60$\\
$30$ & $15$ & \multicolumn{1}{l}{$T_{\gamma=1}\left(  \delta\right)  $} &
$92.3$ & $2.62$\\
$30$ & $15$ & \multicolumn{1}{l}{$T_{\gamma=2}\left(  \delta\right)  $} &
$91.2$ & $2.68$\\
$30$ & $15$ & \multicolumn{1}{l}{$t\left(  \delta\right)  $} & \fbox{$94.6$} &
$2.48$\\\hline
$30$ & $30$ & \multicolumn{1}{l}{$T_{\gamma=-1}\left(  \delta\right)  $} &
$92.4$ & \fbox{$2.10$}\\
$30$ & $30$ & \multicolumn{1}{l}{$T_{\gamma=-0.5}\left(  \delta\right)  $} &
$92.7$ & $2.14$\\
$30$ & $30$ & \multicolumn{1}{l}{$T_{\gamma=0}\left(  \delta\right)  $} &
$93.0$ & $2.24$\\
$30$ & $30$ & \multicolumn{1}{l}{$T_{\gamma=2/3}\left(  \delta\right)  $} &
$93.0$ & $2.31$\\
$30$ & $30$ & \multicolumn{1}{l}{$T_{\gamma=1}\left(  \delta\right)  $} &
$92.9$ & $2.33$\\
$30$ & $30$ & \multicolumn{1}{l}{$T_{\gamma=2}\left(  \delta\right)  $} &
$92.1$ & $2.39$\\
$30$ & $30$ & \multicolumn{1}{l}{$t\left(  \delta\right)  $} & \fbox{$94.1$} &
$2.16$\\\hline
$30$ & $60$ & \multicolumn{1}{l}{$T_{\gamma=-1}\left(  \delta\right)  $} &
$92.0$ & \fbox{$1.92$}\\
$30$ & $60$ & \multicolumn{1}{l}{$T_{\gamma=-0.5}\left(  \delta\right)  $} &
$92.5$ & $1.97$\\
$30$ & $60$ & \multicolumn{1}{l}{$T_{\gamma=0}\left(  \delta\right)  $} &
$92.8$ & $2.03$\\
$30$ & $60$ & \multicolumn{1}{l}{$T_{\gamma=2/3}\left(  \delta\right)  $} &
$92.8$ & $2.11$\\
$30$ & $60$ & \multicolumn{1}{l}{$T_{\gamma=1}\left(  \delta\right)  $} &
$92.8$ & $2.12$\\
$30$ & $60$ & \multicolumn{1}{l}{$T_{\gamma=2}\left(  \delta\right)  $} &
$91.9$ & $2.18$\\
$30$ & $60$ & \multicolumn{1}{l}{$t\left(  \delta\right)  $} & \fbox{$93.5$} &
$1.98$\\\hline
$60$ & $30$ & \multicolumn{1}{l}{$T_{\gamma=-1}\left(  \delta\right)  $} &
$92.9$ & \fbox{$1.71$}\\
$60$ & $30$ & \multicolumn{1}{l}{$T_{\gamma=-0.5}\left(  \delta\right)  $} &
$93.3$ & $1.78$\\
$60$ & $30$ & \multicolumn{1}{l}{$T_{\gamma=0}\left(  \delta\right)  $} &
$93.5$ & $1.83$\\
$60$ & $30$ & \multicolumn{1}{l}{$T_{\gamma=2/3}\left(  \delta\right)  $} &
$93.6$ & $1.87$\\
$60$ & $30$ & \multicolumn{1}{l}{$T_{\gamma=1}\left(  \delta\right)  $} &
$93.5$ & $1.91$\\
$60$ & $30$ & \multicolumn{1}{l}{$T_{\gamma=2}\left(  \delta\right)  $} &
$93.0$ & $1.95$\\
$60$ & $30$ & \multicolumn{1}{l}{$t\left(  \delta\right)  $} & \fbox{$94.5$} &
$1.77$\\\hline
$60$ & $60$ & \multicolumn{1}{l}{$T_{\gamma=-1}\left(  \delta\right)  $} &
$93.6$ & \fbox{$1.51$}\\
$60$ & $60$ & \multicolumn{1}{l}{$T_{\gamma=-0.5}\left(  \delta\right)  $} &
$93.9$ & $1.55$\\
$60$ & $60$ & \multicolumn{1}{l}{$T_{\gamma=0}\left(  \delta\right)  $} &
$94.1$ & $1.59$\\
$60$ & $60$ & \multicolumn{1}{l}{$T_{\gamma=2/3}\left(  \delta\right)  $} &
$94.2$ & $1.65$\\
$60$ & $60$ & \multicolumn{1}{l}{$T_{\gamma=1}\left(  \delta\right)  $} &
$94.2$ & $1.66$\\
$60$ & $60$ & \multicolumn{1}{l}{$T_{\gamma=2}\left(  \delta\right)  $} &
$93.6$ & $1.70$\\
$60$ & $60$ & \multicolumn{1}{l}{$t\left(  \delta\right)  $} & \fbox{$94.6$} &
$1.54$\\\hline
\end{tabular}
\ \ \ \ $%
\caption{Simulated coverage probability and expected width of $0.95$ level CIs of $\delta$ for two pupulations.\label{table1}}%
\end{table}%
%

\begin{figure}[htbp]   \centering
\begin{tabular}
[c]{l}%
\raisebox{-0cm}{\includegraphics[
trim=1.171155in 0.000000in 0.000000in -0.189081in,
height=9.8013cm,
width=18.3045cm
]%
{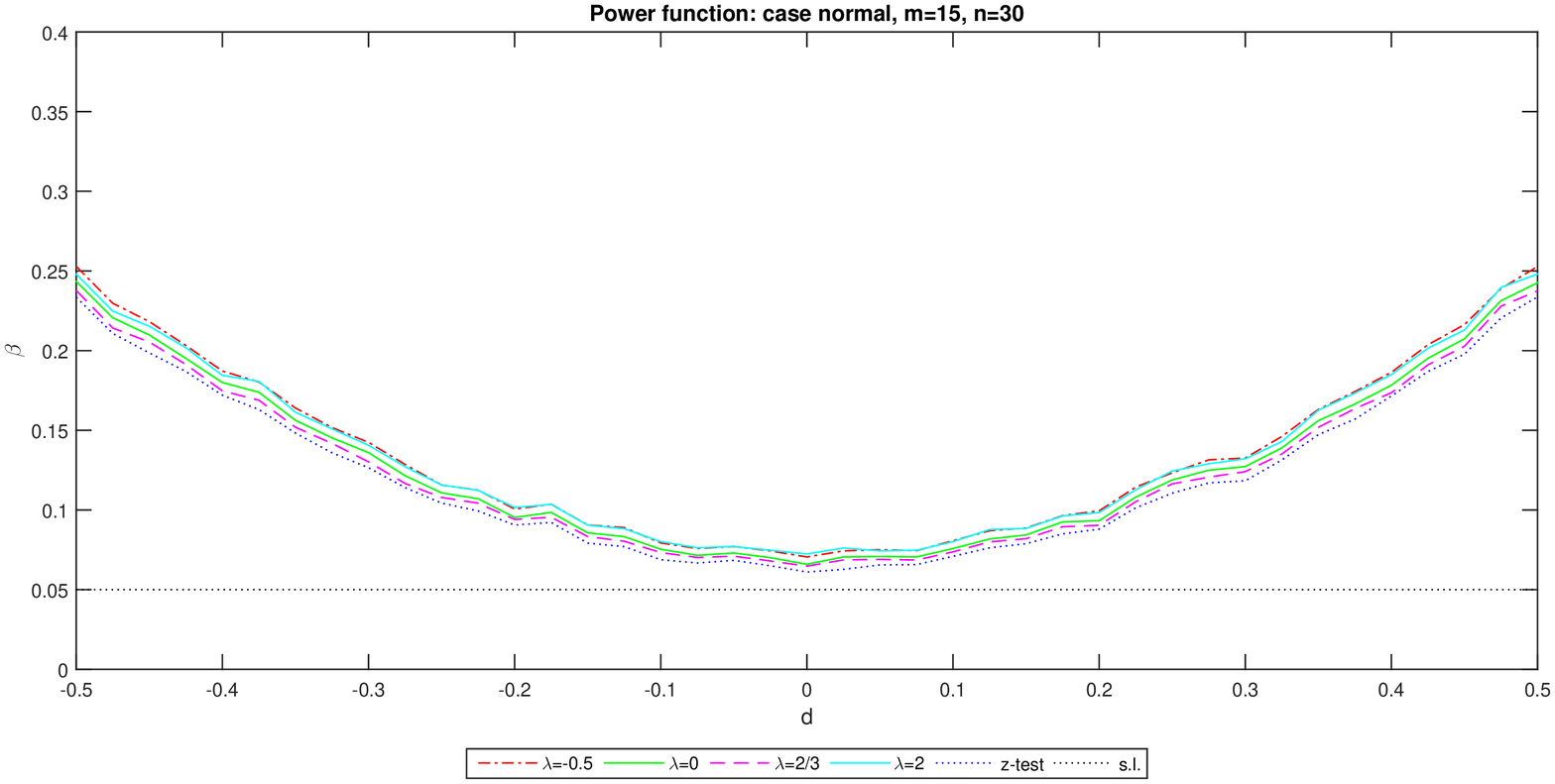}%
}%
\\%
\raisebox{-0cm}{\includegraphics[
trim=1.171155in 0.000000in 0.000000in -0.189081in,
height=9.8013cm,
width=18.3045cm
]%
{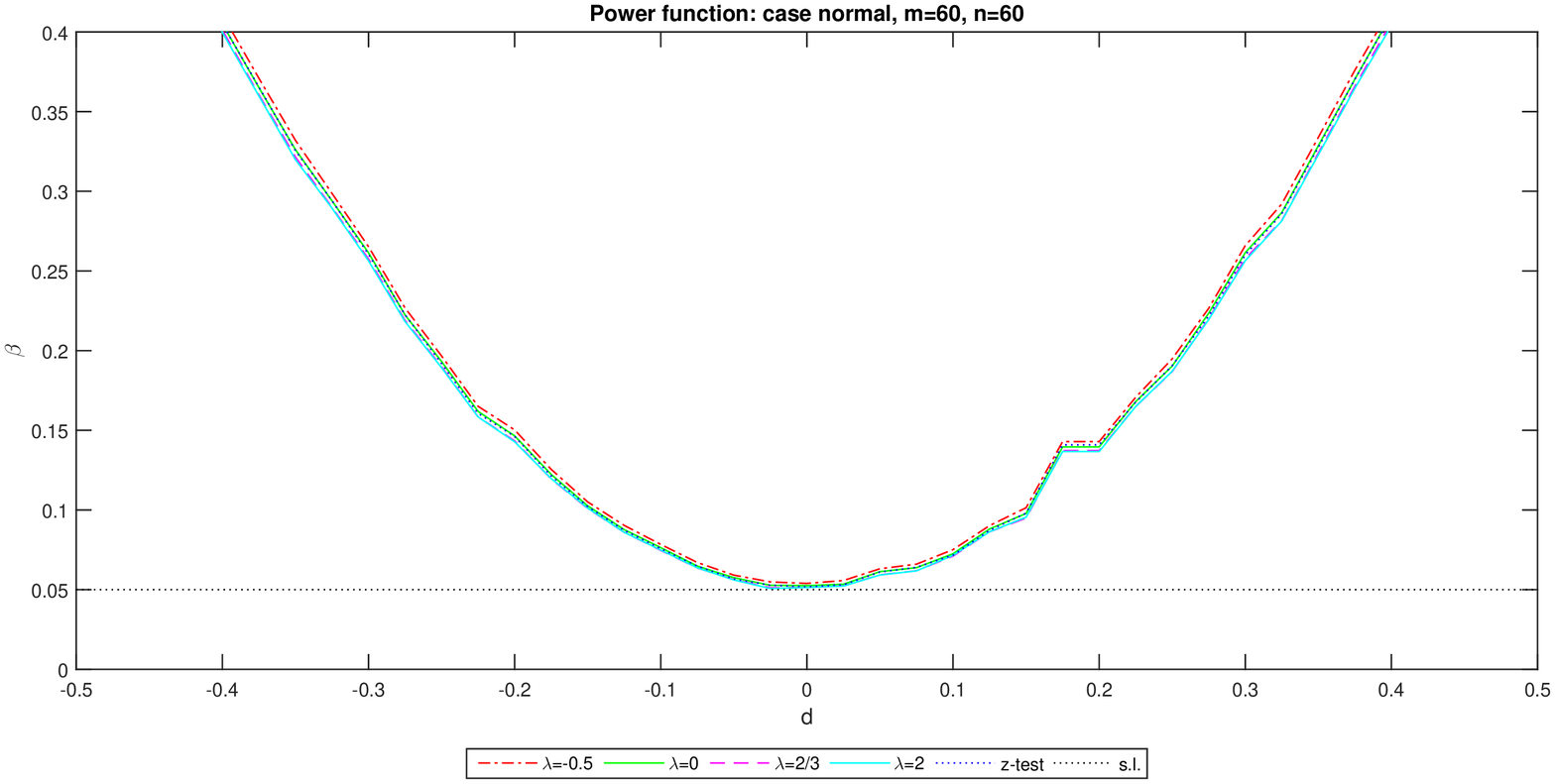}%
}%
\end{tabular}
\caption{Power functions for two normal pupulations.\label{figure1}}%
\end{figure}%
%

\begin{figure}[htbp]   \centering
\begin{tabular}
[c]{l}%
\raisebox{-0cm}{\includegraphics[
trim=1.171155in 0.000000in 0.000000in -0.189081in,
height=9.8035cm,
width=18.3023cm
]%
{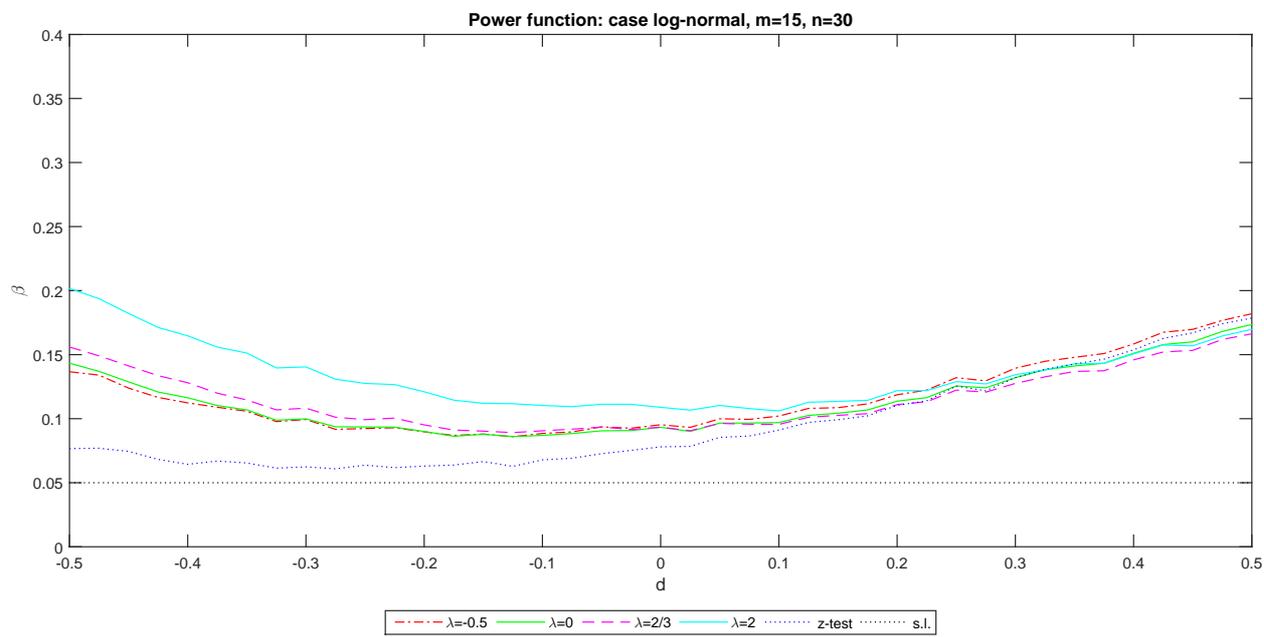}%
}%
\\%
\raisebox{-0cm}{\includegraphics[
trim=1.171155in 0.000000in 0.000000in -0.187821in,
height=9.8013cm,
width=18.3023cm
]%
{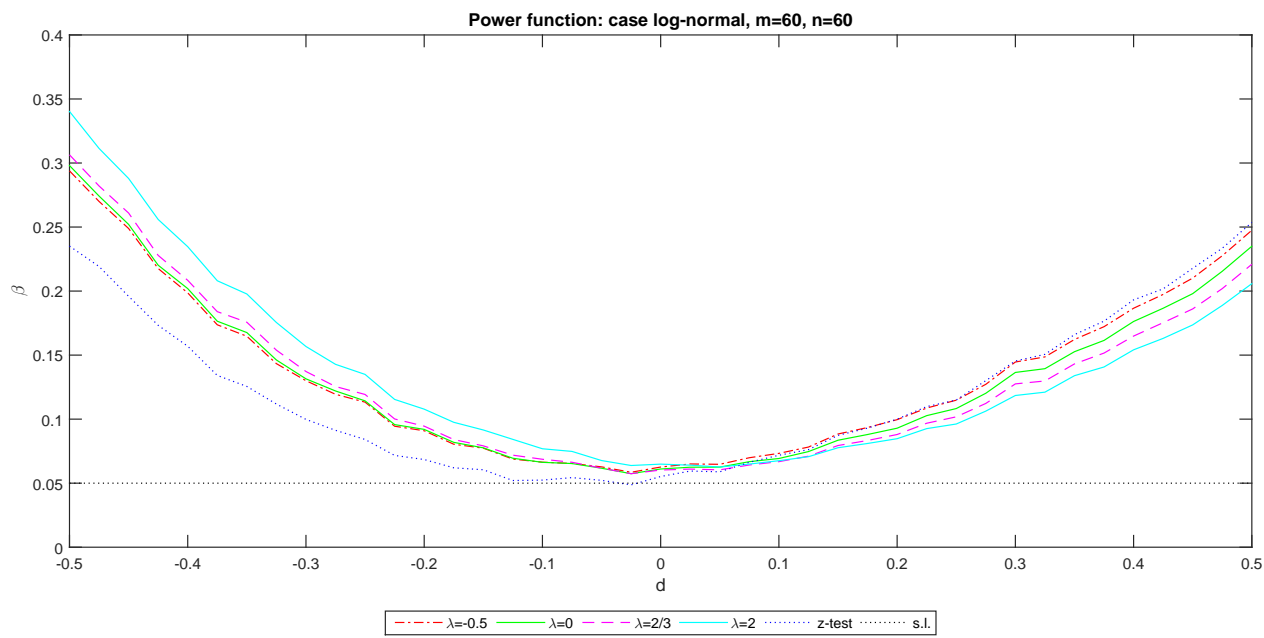}%
}%
\end{tabular}
\caption{Power functions for two log-normal pupulations.\label{figure2}}%
\end{figure}%

\newpage

\section{Numerical Example\label{sec3}}

Yu et al. (2002) presented a data set on evaluating gasoline quality based on
what is known as Reid vapor pressure, collected by the Environmental
Protection Agency of the United States. Two types of Reid vapor pressure
measurements $X$ and $Y$ are included in the data set. Values of $X$ are
obtained by an Agency inspector who visits gas pumps in a city, takes samples
of gasoline of a particular brand, and measures the Reid vapor pressure right
on the spot; values of $Y$, on the other hand, are produced by shipping
gasoline samples to the laboratory for measurements of presumably higher
precision at a high cost. The original data set has a double sampling
structure, with a subset of the sample units having measurements on both $X$
and $Y$. Table \ref{table2} contains two independent samples of a new
reformulated gasoline, one related to $X$ with sample size 30 and the other,
to $Y$ with sample size 15.%

\begin{table}[htbp]  \renewcommand{\arraystretch}{0.84}
\tabcolsep2.8pt  \centering
$%
\begin{tabular}
[c]{ccccccccccc}\hline
$X$\text{ (Field)} & $8.09$ & $8.46$ & $7.37$ & $8.80$ & $7.59$ & $8.62$ &
$7.88$ & $7.98$ & $7.47$ & $8.90$\\
& $8.51$ & $8.69$ & $7.93$ & $7.96$ & $7.45$ & $8.02$ & $7.32$ & $7.45$ &
$7.86$ & $7.88$\\
& $7.39$ & $8.03$ & $7.31$ & $7.44$ & $7.95$ & $7.92$ & $7.53$ & $8.01$ &
$7.16$ & $7.31$\\\hline
$Y$\text{ (Lab)} & $8.28$ & $8.63$ & $9.28$ & $7.85$ & $8.62$ & $9.14$ &
$7.86$ & $7.90$ & $8.52$ & $7.92$\\
& $7.89$ & $8.48$ & $7.95$ & $8.32$ & $7.60$ &  &  &  &  & \\\hline
\end{tabular}
\ \ \ \ $%
\caption{Field and lab data on Reid vapor pressure for newly reformulated gasoline.\label{table2}}%
\end{table}%

One of the assumptions of Yu et al. (2002) is that the field measurement $X$
and the lab measurement $Y$ have common mean $%
\mu
$. The two types of measurements differ, however, in terms of precision. Yu et
al. (2002) also assumed that $(X$, $Y)$ was bivariate normal, which would not
be required under our proposed empirical likelihood approach. In Tsao and Wu
(2006) this example was studied on the basis of the empirical log-likelihood
ratio test. The $95\%$ CIs of $\delta$\ based on $T\in\{t\left(
\delta\right)  \}\cup\{T_{\gamma}\left(  \delta\right)  \}_{\gamma\in\Lambda}$
are summarized in Table \ref{table3}. As in the simulation study, the
narrowest CI width is obtained with $T_{\gamma=-1}\left(  \delta\right)  $. In
all the test-statistics used to construct the CIs $\delta_{0}=0$ is not
contained, so the null hypothesis of equal means is rejected with $0.05$
significance level.%

\begin{table}[htbp]  \renewcommand{\arraystretch}{0.84}
\tabcolsep2.8pt  \centering
$%
\begin{tabular}
[c]{cccc}\hline
$CI$ & lower bound & upper bound & width\\\hline
\multicolumn{1}{l}{$T_{\gamma=-1}\left(  \delta\right)  $} & $0.122$ & $0.703$
& $0.581$\\
\multicolumn{1}{l}{$T_{\gamma=-0.5}\left(  \delta\right)  $} & $0.121$ &
$0.712$ & $0.591$\\
\multicolumn{1}{l}{$T_{\gamma=0}\left(  \delta\right)  $} & $0.121$ & $0.718$
& $0.598$\\
\multicolumn{1}{l}{$T_{\gamma=2/3}\left(  \delta\right)  $} & $0.123$ &
$0.724$ & $0.602$\\
\multicolumn{1}{l}{$T_{\gamma=1}\left(  \delta\right)  $} & $0.124$ & $0.726$
& $0.601$\\
\multicolumn{1}{l}{$T_{\gamma=2}\left(  \delta\right)  $} & $0.133$ & $0.725$
& $0.592$\\
\multicolumn{1}{l}{$t\left(  \delta\right)  $} & $0.101$ & $0.712$ &
$0.611$\\\hline
\end{tabular}
\ \ $%
\caption{Power divergence and z-test based $0.95$ level CIs for field and lab data on Reid vapor pressure for newly reformulated gasoline.\label{table3}}%
\end{table}%

\section{Further extensions\label{sec5}}

\subsection{Extension of the dimension for the random variable}

Let $\boldsymbol{X}_{1},...,\boldsymbol{X}_{m}$ and $\boldsymbol{Y}%
_{1},...,\boldsymbol{Y}_{n}$ be two mutually independent random samples with
common distribution function $F$ and $G$ respectively. Assuming that
$\boldsymbol{X}_{i}$ and $\boldsymbol{Y}_{j}$ take values in $\mathbb{R}^{k}$
and
\begin{align*}
E\left[  \boldsymbol{X}_{i}\right]   &  =\boldsymbol{\mu}_{1},\text{
}Cov\left[  \boldsymbol{X}_{i}\right]  =\boldsymbol{\Sigma}_{1}\text{,
}i=1,...,m,\\
E\left[  \boldsymbol{Y}_{j}\right]   &  =\boldsymbol{\mu}_{2},\text{
}Cov\left[  \boldsymbol{Y}_{j}\right]  =\boldsymbol{\Sigma}_{2}\text{,
}j=1,...,n,
\end{align*}
with $\boldsymbol{\mu}_{1}=\boldsymbol{\mu}$ and $\boldsymbol{\mu}%
_{2}=\boldsymbol{\mu}+\boldsymbol{\delta}$, our interest is in testing
\begin{equation}
H_{0}\text{: }\boldsymbol{\delta}=\boldsymbol{\delta}_{0}\text{ vs. }%
H_{1}\text{: }\boldsymbol{\delta}\neq\boldsymbol{\delta}_{0}, \label{5.1}%
\end{equation}
where $\boldsymbol{\delta}_{0}\in$ $\mathbb{R}^{k}$ and known.

The empirical likelihood under $H_{0}$ is%
\[
\mathcal{L}\left(  \boldsymbol{\delta}\right)  =%
{\textstyle\prod\limits_{i=1}^{m}}
{\textstyle\prod\limits_{j=1}^{n}}
p_{i}q_{j}\text{ s.t. }%
{\textstyle\sum\limits_{i=1}^{m}}
p_{i}(\boldsymbol{x}_{i}-\boldsymbol{\mu})=\boldsymbol{0}_{k},\text{ }%
{\textstyle\sum\limits_{i=1}^{m}}
p_{i}=1,\text{ }%
{\textstyle\sum\limits_{j=1}^{n}}
q_{j}(\boldsymbol{y}_{j}-\boldsymbol{\mu}-\boldsymbol{\delta})=\boldsymbol{0}%
_{k},\text{ }%
{\textstyle\sum\limits_{j=1}^{n}}
q_{j}=1\text{,}%
\]
and in the whole parameter space,%
\[
\mathcal{L}\left(  \boldsymbol{p},\boldsymbol{q}\right)  =%
{\textstyle\prod\limits_{i=1}^{m}}
{\textstyle\prod\limits_{j=1}^{n}}
p_{i}q_{j}:%
{\textstyle\sum\limits_{i=1}^{m}}
p_{i}=%
{\textstyle\sum\limits_{j=1}^{n}}
q_{j}=1,\text{ }p_{i},q_{j}\geq0,
\]
with $p_{i}=F\left(  \boldsymbol{x}_{i}\right)  -F\left(  \boldsymbol{x}%
_{i}^{-}\right)  $ , $q_{j}=G(\boldsymbol{y}_{j})-G(\boldsymbol{y}_{j}^{-})$
and $\boldsymbol{p=}\left(  p_{1},...,p_{m}\right)  ^{T}$, $\boldsymbol{q}%
=\left(  q_{1},...,q_{n}\right)  ^{T}$. The empirical log-likelihood ratio
statistic for testing (\ref{5.1}) is given by
\[
\mathcal{\ell}\left(  \boldsymbol{\delta}_{0}\right)  =-2\log\frac
{\sup_{\boldsymbol{p},\boldsymbol{q}}\mathcal{L}\left(  \boldsymbol{\delta
}_{0}\right)  }{\sup_{\boldsymbol{p},\boldsymbol{q}}\mathcal{L}\left(
\boldsymbol{p},\boldsymbol{q}\right)  }.
\]
Based on Lagrange multiplier methods, $\sup_{\boldsymbol{p},\boldsymbol{q}%
}\mathcal{L}\left(  \boldsymbol{\delta}_{0}\right)  $ is obtained for
\begin{equation}
p_{i}=\frac{1}{m}\frac{1}{1+\boldsymbol{\lambda}_{1}^{T}(\boldsymbol{x}%
_{i}-\boldsymbol{\mu})},\text{ }i=1,...,m, \label{5.2}%
\end{equation}%
\begin{equation}
q_{j}=\frac{1}{n}\frac{1}{1+\boldsymbol{\lambda}_{2}^{T}(\boldsymbol{y}%
_{j}-\boldsymbol{\mu}-\boldsymbol{\delta}_{0})},\text{ }j=1,...,n, \label{5.3}%
\end{equation}
where $m\boldsymbol{\lambda}_{1}^{T}+n\boldsymbol{\lambda}_{2}^{T}=0$. The
empirical maximum likelihood estimates $\widetilde{\boldsymbol{\lambda}}_{1}$,
$\widetilde{\boldsymbol{\lambda}}_{2}$ and $\widetilde{\boldsymbol{\mu}}$ of
$\boldsymbol{\lambda}_{1}$, $\boldsymbol{\lambda}_{2}$ and $\boldsymbol{\mu}$,
under $H_{0}$, can be obtained as the solution of%
\[
\left\{
\begin{array}
[c]{l}%
\frac{1}{m}%
{\textstyle\sum\limits_{i=1}^{m}}
\frac{(\boldsymbol{x}_{i}-\boldsymbol{\mu})}{1+\boldsymbol{\lambda}_{1}%
^{T}(\boldsymbol{x}_{i}-\boldsymbol{\mu})}=\boldsymbol{0}_{k}\\
\frac{1}{n}%
{\textstyle\sum\limits_{j=1}^{n}}
\frac{\left(  \boldsymbol{y}_{j}-\boldsymbol{\mu}-\boldsymbol{\delta}\right)
}{1+\boldsymbol{\lambda}_{2}^{T}\left(  \boldsymbol{y}_{j}-\boldsymbol{\mu
}-\boldsymbol{\delta}\right)  }=\boldsymbol{0}_{k}\\
m\boldsymbol{\lambda}_{1}^{T}+n\boldsymbol{\lambda}_{2}^{T}=\boldsymbol{0}_{k}%
\end{array}
\right.  .
\]
On the other hand $\sup_{\boldsymbol{p},\boldsymbol{q}}\mathcal{L}\left(
\boldsymbol{p},\boldsymbol{q}\right)  $ is obtained for
\begin{equation}
p_{i}=\frac{1}{m}\text{, }i=1,...,m\text{ and }q_{j}=\frac{1}{n}\text{,
}j=1,...,n. \label{5.4}%
\end{equation}
After some algebra, we obtain
\begin{align}
\mathcal{\ell(}\boldsymbol{\delta}_{0})  &  =-2\log\frac{\sup\mathcal{L}%
\left(  \delta_{0}\right)  }{\sup\mathcal{L}\left(  \boldsymbol{p}%
,\boldsymbol{q}\right)  }=-2\left(  -m\log m-%
{\textstyle\sum\limits_{i=1}^{m}}
\log\left(  1+\widetilde{\boldsymbol{\lambda}}_{1}^{T}(\boldsymbol{X}%
_{i}-\boldsymbol{\mu})\right)  -n\log n\right. \nonumber\\
&  \left.  -%
{\textstyle\sum\limits_{j=1}^{n}}
\log\left(  1+\widetilde{\boldsymbol{\lambda}}_{2}^{T}\left(  \boldsymbol{Y}%
_{j}-\boldsymbol{\mu}-\boldsymbol{\delta}_{0}\right)  \right)  +m\log m+n\log
n\right) \nonumber\\
&  =2\left(
{\textstyle\sum\limits_{i=1}^{m}}
\log\left(  1+\widetilde{\boldsymbol{\lambda}}_{1}^{T}(\boldsymbol{X}%
_{i}-\boldsymbol{\mu})\right)  +%
{\textstyle\sum\limits_{j=1}^{n}}
\log\left(  1+\widetilde{\boldsymbol{\lambda}}_{2}^{T}\left(  \boldsymbol{Y}%
_{j}-\boldsymbol{\mu}-\boldsymbol{\delta}\right)  \right)  \right)  .
\label{5.5}%
\end{align}
Under some regularity conditions, it follows that
\[
\Pr\left(  \mathcal{\ell}\left(  \boldsymbol{\delta}_{0}\right)
<\chi_{k,\alpha}^{2}\right)  =\alpha+O(n^{-1}),
\]
where $\chi_{k,\alpha}^{2}$ is the $\alpha$-th order quantile of the $\chi
_{k}^{2}$ distribution.

Let
\[
\widetilde{\boldsymbol{P}}=\left(  \widetilde{p}_{1}\frac{m}{N}%
,...,\widetilde{p}_{m}\nu,\widetilde{q}_{1}(1-\frac{m}{N}),...,\widetilde{q}%
_{n}(1-\frac{m}{N})\right)  ^{T}%
\]
be the estimate the probability vector%
\[
\boldsymbol{P}=\left(  p_{1}\nu,...,p_{m}\nu,q_{1}(1-\nu),...,q_{n}%
(1-\nu)\right)  ^{T},
\]
where $\widetilde{p}_{i}$ and $\widetilde{q}_{j}$ are obtained from
(\ref{5.2}) and (\ref{5.3}) replacing $\boldsymbol{\lambda}_{1}$,
$\boldsymbol{\lambda}_{2}$ and $\boldsymbol{\mu}$ by
$\widetilde{\boldsymbol{\lambda}}_{1}$, $\widetilde{\boldsymbol{\lambda}}_{2}$
and $\widetilde{\boldsymbol{\mu}}$, respectively. In this $k$-dimensional
case, the Kullback-Leibler divergence between the probability vectors
$\boldsymbol{U}$ and $\widetilde{\boldsymbol{P}}$ is given by
\begin{align*}
D_{Kullback}(\boldsymbol{U},\widetilde{\boldsymbol{P}})  &  =%
{\textstyle\sum\limits_{i=1}^{m}}
\frac{1}{N}\log\frac{\frac{1}{N}}{\widetilde{p}_{i}\frac{m}{N}}+%
{\textstyle\sum\limits_{j=1}^{n}}
\frac{1}{N}\log\frac{\frac{1}{N}}{\widetilde{q}_{j}\frac{n}{N}}\\
&  =\frac{1}{N}\left\{
{\textstyle\sum\limits_{i=1}^{m}}
\log\left(  1+\widetilde{\boldsymbol{\lambda}}_{1}^{T}(\boldsymbol{X}%
_{i}-\boldsymbol{\mu})\right)  +%
{\textstyle\sum\limits_{j=1}^{n}}
\log\left(  1+\widetilde{\boldsymbol{\lambda}}_{2}^{T}\left(  \boldsymbol{Y}%
_{j}-\boldsymbol{\mu}-\boldsymbol{\delta}\right)  \right)  \right\}  .
\end{align*}
Therefore, the relationship between $\mathcal{\ell}\left(  \boldsymbol{\delta
}_{0}\right)  $ and the Kullback-Leibler divergence is%
\begin{equation}
\mathcal{\ell}\left(  \boldsymbol{\delta}_{0}\right)  =2ND_{Kullback}%
(\boldsymbol{U},\widetilde{\boldsymbol{P}}). \label{5.6}%
\end{equation}
Based on (\ref{5.6}) the family of empirical phi-divergence test statistics
are defined as%
\[
T_{\phi}\left(  \boldsymbol{\delta}_{0}\right)  =\frac{2N}{\phi^{\prime\prime
}(1)}D_{\phi}(\boldsymbol{U},\widetilde{\boldsymbol{P}}),
\]
with
\[
D_{\phi}(\boldsymbol{U},\widetilde{\boldsymbol{P}})=%
{\textstyle\sum\limits_{i=1}^{m}}
\widetilde{p}_{i}\frac{m}{m+n}\phi\left(  \frac{\frac{1}{N}}{\widetilde{p}%
_{i}\frac{m}{m+n}}\right)  +%
{\textstyle\sum\limits_{j=1}^{n}}
\frac{n}{m+n}\widetilde{q}_{j}\phi\left(  \frac{\frac{1}{N}}{\frac{n}%
{m+n}\widetilde{q}_{j}}\right)  .
\]
Therefore the expression of $T_{\phi}\left(  \boldsymbol{\delta}_{0}\right)  $
is%
\begin{align}
T_{\phi}\left(  \boldsymbol{\delta}_{0}\right)   &  =\frac{2}{\phi
^{\prime\prime}(1)}\left\{
{\textstyle\sum\limits_{i=1}^{m}}
m\widetilde{p}_{i}\phi\left(  \frac{1}{m\widetilde{p}_{i}}\right)  +%
{\textstyle\sum\limits_{j=1}^{n}}
n\widetilde{q}_{j}\phi\left(  \frac{1}{n\widetilde{q}_{j}}\right)  \right\}
\nonumber\\
&  =\frac{2}{\phi^{\prime\prime}(1)}\left\{
{\textstyle\sum\limits_{i=1}^{m}}
\frac{1}{1+\widetilde{\boldsymbol{\lambda}}_{1}^{T}(\boldsymbol{X}%
_{i}-\boldsymbol{\mu})}\phi\left(  1+\widetilde{\boldsymbol{\lambda}}_{1}%
^{T}(\boldsymbol{X}_{i}-\boldsymbol{\mu})\right)  \right. \nonumber\\
&  \left.  +%
{\textstyle\sum\limits_{j=1}^{n}}
\frac{1}{1+\widetilde{\boldsymbol{\lambda}}_{2}^{T}\left(  \boldsymbol{Y}%
_{j}-\boldsymbol{\mu}-\boldsymbol{\delta}\right)  }\phi\left(
1+\widetilde{\boldsymbol{\lambda}}_{2}^{T}\left(  \boldsymbol{Y}%
_{j}-\boldsymbol{\mu}-\boldsymbol{\delta}_{0}\right)  \right)  \right\}  .
\label{5.7}%
\end{align}

A result similar to the one given in Lemma 1 for the $k$-dimensional case is
\[
\widetilde{\boldsymbol{\mu}}=\left(  m\boldsymbol{\Sigma}_{1}^{-1}%
+n\boldsymbol{\Sigma}_{2}^{-1}\right)  ^{-1}\left(  m\boldsymbol{\Sigma}%
_{1}^{-1}\overline{\boldsymbol{X}}+n\boldsymbol{\Sigma}_{2}^{-1}\left(
\overline{\boldsymbol{Y}}-\boldsymbol{\delta}\right)  \right)  +O_{p}%
(\boldsymbol{1}_{k}),
\]
where $\overline{\boldsymbol{X}}=\frac{1}{m}\sum_{i=1}^{m}\boldsymbol{X}_{i}%
$\ and $\overline{\boldsymbol{Y}}=\frac{1}{n}\sum_{j=1}^{n}\boldsymbol{Y}_{i}%
$. Finally, based in this result it is possible to establish%
\[
T_{\phi}\left(  \boldsymbol{\delta}_{0}\right)  \underset{n,m\rightarrow
\infty}{\overset{\mathcal{L}}{\rightarrow}}\chi_{k}^{2}.
\]

\subsection{Extension of the test-statistic using the R\'{e}nyi's divergence}

R\'{e}nyi (1961) introduced the R\'{e}nyi's divergence measure as an extension
of the Kullback-Leibler divergence. Unfortunately this divergence measure is
not a member of the family of phi-divergence measures considered in this
paper. Men\'{e}ndez et al. (1995, 1997) introduced and studied the
(h,phi)-divergence measures in order to have a family of divergence measures
in which the phi-divergence measures as well as the R\'{e}nyi divergence
measure are included. But not only the R\'{e}nyi divergence measure is
included in this new family but another important divergence measures not
include in the family of phi-divergence measures are included. For more
details about the different divergence measures included in the
(h,phi)-divergence see for instance, Pardo (2006). Based on the
(h,phi)-divergence measures between the probability vectors $\boldsymbol{U}$
and $\boldsymbol{P}$, defined in (\ref{9}) and (\ref{10}) respectively, we can
consider the following family of empirical (h,phi)-divergence test statistics
for the two-sample problem considered in (\ref{1})
\begin{equation}
T_{\phi}^{h}\left(  \delta_{0}\right)  =h\left(  T_{\phi}\left(  \delta
_{0}\right)  \right)  ,\label{a}%
\end{equation}
where $h$ is a differentiable increasing function from $\left[  0,\infty
\right)  $ onto $\left[  0,\infty\right)  $ with $h(0)=0$ and $h^{\prime
}(0)>0$. If we consider
\[
h(x)=\frac{1}{a(a-1)}\log\left(  a(a-1)x+1\right)  ,\text{\ }a\neq0,1
\]
in (\ref{a}), and
\[
\phi\left(  x\right)  =\frac{x^{a}-a\left(  x-1\right)  -1}{a(a-1)},\text{
}a\neq0,1,
\]
we get
\begin{align*}
T_{R\acute{e}nyi}\left(  \delta_{0}\right)   &  =T_{\phi}^{h}\left(
\delta_{0}\right)  =\frac{1}{a(a-1)}\log\left(  a(a-1)T_{\phi}\left(
\delta_{0}\right)  +1\right)  \\
&  =\frac{1}{a(a-1)}\log\left\{
{\textstyle\sum\limits_{i=1}^{m}}
\left(  1+\widetilde{\lambda}_{1}\left(  x_{i}-\widetilde{\mu}\right)
\right)  ^{a-1}+%
{\textstyle\sum\limits_{j=1}^{n}}
\left(  1+\widetilde{\lambda}_{2}\left(  y_{j}-\left(  \widetilde{\mu}%
+\delta_{0}\right)  \right)  \right)  ^{a-1}\right\}  ,
\end{align*}
i.e., the empirical R\'{e}nyi's divergence test statistics for testing
(\ref{1}). For $a=0$ and $a=1$, we get
\[
\lim_{a\rightarrow0}T_{R\acute{e}nyi}\left(  \delta_{0}\right)
=2ND_{Kullback}\left(  \boldsymbol{U},\boldsymbol{P}\right)
\]
and
\[
\lim_{a\rightarrow1}T_{R\acute{e}nyi}\left(  \delta_{0}\right)
=2ND_{Kullback}\left(  \boldsymbol{P,U}\right)  .
\]
It is clear that
\[
h\left(  T_{\phi}\left(  \delta_{0}\right)  \right)  =h(0)+h^{\prime
}(0)T_{\phi}\left(  \delta_{0}\right)  +o\left(  T_{\phi}\left(  \delta
_{0}\right)  \right)  .
\]
Therefore%
\[
\frac{T_{\phi}^{h}\left(  \delta_{0}\right)  }{h^{\prime}(0)}%
\underset{n,m\rightarrow\infty}{\overset{\mathcal{L}}{\rightarrow}}\chi
_{1}^{2}.
\]
In the same way can be established for the problem considered in (\ref{5.1})
that
\[
\frac{T_{\phi}^{h}\left(  \boldsymbol{\delta}_{0}\right)  }{h^{\prime}%
(0)}\underset{n,m\rightarrow\infty}{\overset{\mathcal{L}}{\rightarrow}}%
\chi_{k}^{2},
\]
where
\[
T_{\phi}^{h}\left(  \boldsymbol{\delta}_{0}\right)  =h(T_{\phi}\left(
\boldsymbol{\delta}_{0}\right)  )
\]
with $T_{\phi}\left(  \boldsymbol{\delta}_{0}\right)  $ defined in (\ref{5.7}).

\textbf{Acknowledgement. }This research is partially supported by Grants
MTM2012-33740 from Ministerio de Economia y Competitividad (Spain).

\section*{Appendix}

\subsection{Proof of Lemma \ref{Lem}\label{Ap1}}

In a similar way as in Hall and Scala (1990), we can establish
\[
\lambda_{1}=\lambda_{1}\left(  \mu\right)  =-\sigma_{1}^{-2}\left(
\overline{X}-\mu\right)  +O_{p}(m^{-1})\text{ and }\lambda_{2}=\lambda
_{2}\left(  \mu\right)  =-\sigma_{2}^{-2}\left(  \overline{Y}-\mu-\delta
_{0}\right)  +O_{p}(n^{-1}).
\]
Now applying that
\[
m\lambda_{1}\left(  \mu\right)  +n\lambda_{2}\left(  \mu\right)  =0,
\]
we have%
\begin{equation}
m\sigma_{1}^{-2}\left(  \overline{X}-\mu\right)  +O_{p}(1)+n\sigma_{2}%
^{-2}\left(  \overline{Y}-\mu-\delta_{0}\right)  +O_{p}(1)=0. \label{14}%
\end{equation}
Solving the equation for $\mu$ we have the enunciated result.

\subsection{Proof of Theorem \ref{Th1}\label{Ap2}}

First we are going to establish
\begin{equation}
\frac{2}{\phi^{\prime\prime}\left(  1\right)  }%
{\textstyle\sum\limits_{i=1}^{m}}
mp_{i}\phi\left(  \frac{1}{mp_{i}}\right)  =m\left(  \frac{\overline{X}-\mu
}{\sigma_{1}}\right)  ^{2}+o_{p}(1) \label{15}%
\end{equation}%
\begin{equation}
\frac{2}{\phi^{\prime\prime}\left(  1\right)  }%
{\textstyle\sum\limits_{j=1}^{n}}
nq_{j}\phi\left(  \frac{1}{nq_{j}}\right)  =n\left(  \frac{\overline{Y}%
-\mu-\delta_{0}}{\sigma_{2}}\right)  ^{2}+o_{p}(1). \label{16}%
\end{equation}
If we denote $W_{i}=\lambda_{1}\left(  \mu\right)  \left(  X_{i}-\mu\right)  $
we have $\phi\left(  \frac{1}{mp_{i}}\right)  =\phi\left(  1+W_{i}\right)  $.
A Taylor expansion gives
\[
\phi\left(  1+W_{i}\right)  =\phi\left(  1\right)  +\phi^{\prime}\left(
1\right)  W_{i}+\frac{1}{2}\phi^{\prime\prime}\left(  1\right)  W_{i}%
^{2}+o\left(  W_{i}^{2}\right)  .
\]
On the other hand
\[
mp_{i}=\frac{1}{1+\lambda_{1}\left(  \mu\right)  \left(  X_{i}-\mu\right)
}=\frac{1}{1+W_{i}}=1-W_{i}+W_{i}^{2}+o\left(  W_{i}^{2}\right)  .
\]
Then%
\begin{align*}
\frac{2}{\phi^{\prime\prime}\left(  1\right)  }%
{\textstyle\sum\limits_{i=1}^{m}}
mp_{i}\phi\left(  \frac{1}{mp_{i}}\right)   &  =\frac{2}{\phi^{\prime\prime
}\left(  1\right)  }%
{\textstyle\sum\limits_{i=1}^{m}}
\left(  1-W_{i}+W_{i}^{2}+o\left(  W_{i}^{2}\right)  \right)  \left(  \frac
{1}{2}\phi^{\prime\prime}\left(  1\right)  W_{i}^{2}+o\left(  W_{i}%
^{2}\right)  \right) \\
&  =\frac{2}{\phi^{\prime\prime}\left(  1\right)  }\left(  \frac{1}{2}%
\phi^{\prime\prime}\left(  1\right)
{\textstyle\sum\limits_{i=1}^{m}}
W_{i}^{2}+%
{\textstyle\sum\limits_{i=1}^{m}}
o\left(  W_{i}^{2}\right)  -\frac{1}{2}\phi^{\prime\prime}\left(  1\right)
{\textstyle\sum\limits_{i=1}^{m}}
W_{i}^{3}\right. \\
&  \left.  -%
{\textstyle\sum\limits_{i=1}^{m}}
o\left(  W_{i}^{2}\right)  W_{i}+\frac{1}{2}\phi^{\prime\prime}\left(
1\right)
{\textstyle\sum\limits_{i=1}^{m}}
W_{i}^{4}+%
{\textstyle\sum\limits_{i=1}^{m}}
o\left(  W_{i}^{2}\right)  W_{i}^{2}\right. \\
&  \left.  +\frac{1}{2}\phi^{\prime\prime}\left(  1\right)
{\textstyle\sum\limits_{i=1}^{m}}
o\left(  W_{i}^{2}\right)  W_{i}^{2}+%
{\textstyle\sum\limits_{i=1}^{m}}
o\left(  W_{i}^{2}\right)  o\left(  W_{i}^{2}\right)  \right) \\
&  =%
{\textstyle\sum\limits_{i=1}^{m}}
W_{i}^{2}+\frac{2}{\phi^{\prime\prime}\left(  1\right)  }%
{\textstyle\sum\limits_{i=1}^{m}}
o\left(  W_{i}^{2}\right)  -%
{\textstyle\sum\limits_{i=1}^{m}}
W_{i}^{3}-%
{\textstyle\sum\limits_{i=1}^{m}}
o\left(  W_{i}^{2}\right)  W_{i}+%
{\textstyle\sum\limits_{i=1}^{m}}
W_{i}^{4}\\
&  +\frac{2}{\phi^{\prime\prime}\left(  1\right)  }%
{\textstyle\sum\limits_{i=1}^{m}}
o\left(  W_{i}^{2}\right)  +%
{\textstyle\sum\limits_{i=1}^{m}}
o\left(  W_{i}^{2}\right)  W_{i}^{2}+\frac{2}{\phi^{\prime\prime}\left(
1\right)  }%
{\textstyle\sum\limits_{i=1}^{m}}
o\left(  W_{i}^{2}\right)  o\left(  W_{i}^{2}\right)  .
\end{align*}
But\medskip\newline$\bullet$ $\frac{2}{\phi^{\prime\prime}\left(  1\right)  }%
{\textstyle\sum\limits_{i=1}^{m}}
o\left(  W_{i}^{2}\right)  =\frac{2}{\phi^{\prime\prime}\left(  1\right)  }%
{\textstyle\sum\limits_{i=1}^{m}}
o\left(  \lambda_{1}^{2}\left(  \mu\right)  \left(  X_{i}-\mu\right)
^{2}\right)  =\frac{2}{\phi^{\prime\prime}\left(  1\right)  }m\frac{1}{m}%
{\textstyle\sum\limits_{i=1}^{m}}
\left(  X_{i}-\mu\right)  ^{2}o\left(  \lambda_{1}^{2}\left(  \mu\right)
\right)  $\newline$\hspace*{2.75cm}=mo_{p}(1)o\left(  O_{p}\left(
m^{-1}\right)  \right)  =o_{p}(1)$, because
\[
\lambda_{1}\left(  \mu\right)  =O_{p}(m^{-1/2})
\]
(see page 220 in Owen (2001)), and
\[
\frac{1}{m}%
{\textstyle\sum\limits_{i=1}^{m}}
\left(  X_{i}-\mu\right)  ^{2}=o_{p}(1)
\]
applying the strong law of large numbers.\medskip\newline$\bullet$
$\left\vert
{\textstyle\sum\limits_{i=1}^{m}}
W_{i}^{3}\right\vert \leq\left\vert \lambda_{1}\left(  \mu\right)
^{3}\right\vert m\frac{1}{m}%
{\textstyle\sum\limits_{i=1}^{m}}
\left\vert X_{i}-\mu\right\vert ^{3}=O_{p}(m^{-3/2})mo\left(  m^{1/2}\right)
=o_{p}(1)$, because
\[
\frac{1}{m}%
{\textstyle\sum\limits_{i=1}^{m}}
\left\vert X_{i}-\mu\right\vert ^{3}=o\left(  m^{1/2}\right)  ,
\]
by Lemma 11.3 in page 218 in Owen (2001).\medskip\newline$\bullet$
$\left\vert
{\textstyle\sum\limits_{i=1}^{m}}
o\left(  W_{i}^{2}\right)  W_{i}\right\vert \leq o\left(  \lambda_{1}%
^{2}\left(  \mu\right)  \right)  \lambda_{1}\left(  \mu\right)  m\frac{1}{m}%
{\textstyle\sum\limits_{i=1}^{m}}
\left\vert X_{i}-\mu\right\vert ^{3}=o\left(  O_{p}(m^{-1}\right)
O_{p}(m^{-1/2})o(m^{3/2})=o_{p}(1)$.\medskip\newline$\bullet$ $\left\vert
{\textstyle\sum\limits_{i=1}^{m}}
W_{i}^{4}\right\vert \leq\left\vert \lambda_{1}\left(  \mu\right)
^{4}\right\vert
{\textstyle\sum\limits_{i=1}^{m}}
\left\vert X_{i}-\mu\right\vert ^{4}\leq O_{p}\left(  m^{-2}\right)
mZ_{m}\frac{1}{m}%
{\textstyle\sum\limits_{i=1}^{m}}
\left\vert X_{i}-\mu\right\vert ^{3}=O_{p}\left(  m^{-2}\right)
mo(m^{1/2})O(m^{1/2})=o_{p}(1),$ because
\[
Z_{m}=\max_{1\leq i\leq m}\left\vert X_{i}-\mu\right\vert =O(m^{1/2})
\]
applying Lemma 11.2 in page 218 in Owen (2001).\medskip\newline$\bullet$
$\left\vert \frac{2}{\phi^{\prime\prime}\left(  1\right)  }%
{\textstyle\sum\limits_{i=1}^{m}}
o\left(  W_{i}^{2}\right)  o\left(  W_{i}^{2}\right)  \right\vert \leq\frac
{2}{\phi^{\prime\prime}\left(  1\right)  }\left\vert \lambda_{1}\left(
\mu\right)  ^{4}\right\vert
{\textstyle\sum\limits_{i=1}^{m}}
\left\vert X_{i}-\mu\right\vert ^{4}=o_{p}(1)$. Therefore%
\begin{align*}
\frac{2}{\phi^{\prime\prime}\left(  1\right)  }%
{\textstyle\sum\limits_{i=1}^{m}}
mp_{i}\phi\left(  \frac{1}{mp_{i}}\right)   &  =%
{\textstyle\sum\limits_{i=1}^{m}}
W_{i}^{2}+o_{p}(1)\\
&  =%
{\textstyle\sum\limits_{i=1}^{m}}
\lambda_{1}^{2}\left(  \mu\right)  (X_{i}-\mu)^{2}+o_{p}(1)\\
&  =\sigma_{1}^{-4}\left(  \overline{X}-\mu\right)  ^{2}m\frac{1}{m}%
{\textstyle\sum\limits_{i=1}^{m}}
(X_{i}-\mu)^{2}+o_{p}(1)\\
&  =m\left(  \frac{\overline{X}-\mu}{\sigma_{1}}\right)  ^{2}+o_{p}(1).
\end{align*}
In a similar way we can get
\[
\frac{2}{\phi^{\prime\prime}\left(  1\right)  }%
{\textstyle\sum\limits_{j=1}^{n}}
nq_{j}\phi\left(  \frac{1}{nq_{j}}\right)  =n\left(  \frac{\overline{Y}%
-\mu-\delta_{0}}{\sigma_{2}}\right)  ^{2}+o_{p}(1).
\]
Therefore,
\begin{align*}
T_{\phi}\left(  \delta_{0}\right)   &  =\frac{2}{\phi^{\prime\prime}%
(1)}\left\{
{\textstyle\sum\limits_{i=1}^{m}}
m\widetilde{p}_{i}\phi\left(  \frac{1}{m\widetilde{p}_{i}}\right)  +%
{\textstyle\sum\limits_{j=1}^{n}}
n\widetilde{q}_{j}\phi\left(  \frac{1}{n\widetilde{q}_{j}}\right)  \right\} \\
&  =m\left(  \frac{\overline{X}-\widetilde{\mu}}{\sigma_{1}}\right)
^{2}+n\left(  \frac{\overline{Y}-\widetilde{\mu}-\delta_{0}}{\sigma_{2}%
}\right)  ^{2}+o_{p}(1).
\end{align*}
Applying (\ref{14}),
\[
n\sigma_{2}^{-2}\left(  \overline{Y}-\widetilde{\mu}-\delta_{0}\right)
=-m\sigma_{1}^{-2}\left(  \overline{X}-\widetilde{\mu}\right)  +O_{p}(1)
\]
and
\begin{align*}
T_{\phi}\left(  \delta_{0}\right)   &  =m\left(  \frac{\overline
{X}-\widetilde{\mu}}{\sigma_{1}}\right)  ^{2}-m\sigma_{1}^{-2}\left(
\overline{X}-\widetilde{\mu}\right)  \left(  \overline{Y}-\widetilde{\mu
}-\delta_{0}\right)  +o_{p}(1)\\
&  =m\left(  \frac{\overline{X}-\widetilde{\mu}}{\sigma_{1}^{2}}\right)
\left(  \overline{X}-\widetilde{\mu}-\overline{Y}+\widetilde{\mu}+\delta
_{0}\right)  +o_{p}(1)\\
&  =m\left(  \frac{\overline{X}-\widetilde{\mu}}{\sigma_{1}^{2}}\right)
\left(  \overline{X}-\overline{Y}+\delta_{0}\right)  +o_{p}(1)\\
&  =-n\sigma_{2}^{-2}\left(  \overline{Y}-\widetilde{\mu}-\delta_{0}\right)
\left(  \overline{X}-\overline{Y}+\delta_{0}\right)  +o_{p}(1).
\end{align*}
From (\ref{muHat}) we have%
\begin{align*}
\overline{Y}-\mu-\delta_{0}  &  =\overline{Y}-\dfrac{\dfrac{m\overline{X}%
}{\sigma_{1}^{2}}+n\dfrac{\left(  \overline{Y}-\delta_{0}\right)  }{\sigma
_{2}^{2}}}{\dfrac{m}{\sigma_{1}^{2}}+\dfrac{n}{\sigma_{2}^{2}}}-\delta_{0}\\
&  =\frac{m}{\sigma_{1}^{2}}\frac{\left(  -\overline{X}+\overline{Y}%
-\delta_{0}\right)  }{\frac{\sigma_{1}^{2}}{m}+\frac{\sigma_{2}^{2}}{n}}.
\end{align*}
Therefore,%
\begin{align*}
T_{\phi}\left(  \delta_{0}\right)   &  =\frac{nm\sigma_{1}^{-2}\left(
\overline{Y}-\overline{X}-\delta_{0}\right)  }{\frac{\sigma_{1}^{2}}{m}%
+\frac{\sigma_{2}^{2}}{n}}(-1)\left(  \overline{X}-\overline{Y}+\delta
_{0}\right)  \sigma_{2}^{-2}\\
&  =\frac{mn}{m+n}\left(  \overline{X}-\left(  \overline{Y}-\delta_{0}\right)
\right)  ^{2}\sigma_{1}^{-2}\sigma_{2}^{-2}\left(  \frac{m\sigma_{1}%
^{-2}+n\sigma_{2}^{-2}}{m+n}\right)  ^{-1}+o_{p}(1)\\
&  =\frac{1}{\frac{m}{m+n}\sigma_{2}^{2}+\frac{n}{m+n}\sigma_{1}^{2}}\frac
{mn}{m+n}\left(  \overline{X}-\left(  \overline{Y}-\delta_{0}\right)  \right)
^{2}+o_{p}(1).
\end{align*}
Now we have,%
\begin{align*}
&  \sqrt{m}\left(  \overline{X}-\mu\right)  \underset{m\rightarrow
\infty}{\overset{\mathcal{L}}{\rightarrow}}\mathcal{N}(0,\sigma_{1}^{2}),\\
&  \sqrt{n}\left(  \overline{Y}-\left(  \mu-\delta_{0}\right)  \right)
\underset{n\rightarrow\infty}{\overset{\mathcal{L}}{\rightarrow}}%
\mathcal{N}(0,\sigma_{2}^{2}).
\end{align*}
and%
\begin{align*}
&  \sqrt{\frac{mn}{m+n}}\left(  \overline{X}-\mu\right)
\underset{n,m\rightarrow\infty}{\overset{\mathcal{L}}{\rightarrow}}%
\mathcal{N}(0,\left(  1-\nu\right)  \sigma_{1}^{2}).\\
&  \sqrt{\frac{mn}{m+n}}\left(  \overline{Y}-\left(  \mu+\delta_{0}\right)
\right)  \underset{n,m\rightarrow\infty}{\overset{\mathcal{L}}{\rightarrow}%
}\mathcal{N}(0,\nu\sigma_{2}^{2}),
\end{align*}
where is such that (\ref{ass}). Hence
\[
\sqrt{\frac{mn}{m+n}}\left(  \overline{X}-\overline{Y}+\delta_{0}\right)
\underset{n,m\rightarrow\infty}{\overset{\mathcal{L}}{\rightarrow}}%
\mathcal{N}(0,\left(  1-\nu\right)  \sigma_{1}^{2}+\nu\sigma_{2}^{2}),
\]
from which is obtained%
\[
\sqrt{\frac{1}{\frac{m}{m+n}\sigma_{2}^{2}+\frac{n}{m+n}\sigma_{1}^{2}}}%
\sqrt{\frac{mn}{m+n}}\left(  \overline{X}-\overline{Y}+\delta_{0}\right)
\underset{n,m\rightarrow\infty}{\overset{\mathcal{L}}{\rightarrow}}%
\mathcal{N}(0,1)
\]
and now the result follows.
\end{document}